\RequirePackage{fix-cm}

\documentclass{svjour3}                     
\smartqed  

\usepackage{graphicx}
\usepackage{amssymb}
\usepackage{amsmath}
\usepackage{multirow}
\usepackage{slashbox}
\usepackage{cite}
\usepackage{lscape}
\usepackage{xcolor}
\usepackage[linesnumbered,ruled,vlined]{algorithm2e}
\newcommand{\RN}[1]{%
  \textup{\uppercase\expandafter{\romannumeral#1}}%
}

\begin{document}
\title{LP Formulations of sufficient statistic based strategies in Finite Horizon Two-Player Zero-Sum Stochastic Bayesian games
}


\author{Nabiha Nasir Orpa         \and
        Lichun Li 
}


\institute{Nabiha Nasir Orpa and Lichun Li \at
              Department of Industrial and Manufacturing Engineering, FAMU-FSU College of Engineering, 2525 Pottsdamer St, Tallahassee, FL 32310, USA\\
              \email{no18k@my.fsu.edu, lichunli@eng.famu.fsu.edu}
}
\date{Received: date / Accepted: date}
\maketitle
\begin{abstract}
This paper studies two-player zero-sum stochastic Bayesian games where each player has its own dynamic state that is unknown to the other. Generally speaking, players compute their strategies based on the history information which grows exponentially with the time horizon of the game. To save the memory of the players, this paper provides LP formulations to compute the sufficient statistic based optimal strategy and to update a fully accessible sufficient statistic. Because the size of the LPs grows exponentially in time horizon, the LP formulation cannot be directly applied in to cases with long time horizon. To address this problem, we apply a short horizon game repeatedly and analyze the performance. The main results are demonstrated in a security problem of underwater sensor networks.

\end{abstract}
\section{Introduction}
\label{intro}
Because of the multi-agent nature, game theory has great potential in solving or explaining economic, social, and engineering problems. Game theory has been used in addressing AdWord problems \cite{charles2013budget}, enhancing the security of Los Angeles airport \cite{pita2009using}, advising in presidential election and nuclear disarmament \cite{feddersen2006theory,aumann1995repeated}, explaining and anticipating disease spreading \cite{eksin2017disease}, and many other problems. One common property of these problems is that the individuals or agents in the problems have their own private information not shared with the others. For example, didders in AdWord problems may not reveal its budget to the other bidders.

If one or more agents in a game don't have complete information about the game, we call the game a game with incomplete information, which was first introduced in \cite{doi:10.1287/mnsc.1040.0297}. In this case, a player in the game makes its strategy according to its observations like the other players' actions and/or its own payoff. Two-player zero-sum games with incomplete information are special cases of games with incomplete information, and the focus of this paper.

This paper studies two player zero-sum stochastic games with incomplete information on both sides, which are also called stochastic Bayesian games. In these games, both players have their own types (private information) which will change stage by stage, and the payoffs depend on both players' types. We assume that both players can observe the actions of each other, and the one-stage payoffs won't be revealed until the end of the game.

The existing literature \cite{aumann1995repeated,2226507620060801, Rosenberg1998,sorin2003stochastic,de1996repeated,sorin2002first,zamir1992repeated,mertens1971value, 8567999,li2019efficient,rosenberg2004stochastic,gensbittel2015value} had rich descriptive results about Bayeisan two-player zero-sum games. \cite{aumann1995repeated, de1996repeated, de1999cav, mertens1994repeated, mertens1971value, sorin2002first, zamir1992repeated} showed the recursive formula and the existence of game value in repeated games, and \cite{gensbittel2015value, 2226507620060801, Rosenberg1998, rosenberg2004stochastic, sorin2003stochastic} extended the results to stochastic games. While most work directly studied the primal games, \cite{de1996repeated, de1999cav, Rosenberg1998} studied the dual games and provided fully accessible sufficient statistics in the dual games. It was shown that with some special initial parameter, the optimal strategy in the dual game is the optimal strategy in the primal game.

This paper focuses on prescriptive results, i.e., computationally workable strategies in the primal game. With finite horizon, the Bayesian stochastic game can be transferred to a finite game tree, and sequence form \cite{von1996efficient} can be used to develop LP formulation to compute the optimal strategy. This strategy is based on action history which will grow exponentially with respect to the time horizon. To save the memory of the player, this paper provides LP formulations for sufficient statistic based optimal strategy. According to the previous results in the literature, we need to find the special initial sufficient statistic in the dual game, compute the optimal strategy in the dual game, and then update the sufficient statistic for the next step. While the LP formulation in \cite{von1996efficient} can  compute the optimal strategy in the dual game with some adjustment, we need to know how to compute the initial sufficient statistic and update the sufficient statistic in the dual games. To solve this problems, we analyze the LP formulation in the primal game explicitly, and construct an LP to compute both the initial sufficient statistic of the dual game and the optimal strategy of the primal game simultaneously. Based on this LP, we further develop the LPs in the dual games to compute the optimal strategy and to update the sufficient statistics.

Using these LPs directly for long Bayesian games is a computationally heavy approach as the size of the LPs grows exponentially with respect to the time horizon of the game. To address this issue we propose an algorithm which divides the game into multiple small windows and compute the optimal strategies of the players for each window. We analyze the performance of this window-by-window method, and show that the difference between the worst case performance of the window-by-window method and the game value is bounded.

The rest of the paper is organized as follows. The structure and parameters of the game are described in section \ref{Game model}. The concept of primal and dual games with their properties is provided in section \ref{preliminary results}. The motivation and the problem statement of this paper is in section \ref{Problem Statement}. Section \ref{LPs for sufficient statistics based strategies} provides the LPs and sufficient statistic based algorithm to compute the optimal strategies of the players. The performance bound of this algorithm is provided in section \ref{performance analysis}. Section \ref{case study} computes the optimal strategies of the players using the sufficient statistic based algorithm for underwater acoustic sensor network jamming problem and shows the results satisfy the performance bound. 
\section{Game model}
\label{Game model}
Let $\mathbb{R}^n$ denotes the n-dimensional real space, and $\mathcal{K}$ be a finite set. The cardinality of $\mathcal{K}$ and the set of probability distribution over $\mathcal{K}$ is denoted by $|\mathcal{K}|$ and $\Delta(\mathcal{K})$, respectively. A two-player zero-sum stochastic game with incomplete information on both sides is specified by the nine-tuple $\mathcal{(K,L,}\mathcal{A,B},p,q,
P,Q,G)$, where
\begin{itemize}
  \item $\mathcal{K}$ and $\mathcal{L}$ are non-empty finite sets, called player $1$ and $2$'s state sets, respectively.
  \item $\mathcal{A}$ and $\mathcal{B}$ are non-empty finite sets, called player $1$ and $2$'s action sets, respectively where $a_t \in \mathcal{A}$ is the action of player $1$ and $b_t \in \mathcal{B}$ is the action of player $2$ at stage $t$.
  \item $G_{k,l} \in {\Bbb R^{\mathcal{|A|}\times \mathcal{|B|}}}$ is the payoff matrix given player $1$'s state $k \in \mathcal{K}$ and player $2$'s state $ l \in \mathcal{L} $.All the elements of $G$ are non negative. The element $ G_{k,l}(a,b) $ is player $1$'s one stage payoff or player $2$'s one stage cost if the state of player $1$ and $2$ are $ k $ and $ l $, respectively and the current action of player $1$ and 2 are $ a $ and $ b $, respectively.
  \item $ p \in \Delta(\mathcal{K})$ is the initial probability of $\mathcal{K}$ and $q \in \Delta(\mathcal{L}) $ is the initial probability of $\mathcal{L}$.
  \item $P \in {\Bbb R}^{|K|\times |K|} $ and $Q \in {\Bbb R}^{|L|\times |L|}$ are player $1$ and $2$'s transition matrices, respectively. $P_{a_{t},b_{t}}(k_{t},k_{t+1})$ is a conditional probability of the next state of player $1$, $k_{t+1}$, given the current actions of both players and current state of player 1. Similarly, $Q_{a_{t},b_{t}}(l_{t},l_{t+1})$ is a conditional probability of the next state of player $2$, $l_{t+1}$, given the current actions of both players and current state of player $2$.
\end{itemize}

The nine-tuple $\mathcal{(K,L,A,B},p,q,P,Q,G)$ is the common knowledge of the two players. At stage $t=1$, the initial state of player 1 and player 2 are chosen independently by nature according to the initial probability $ p$ and $ q $, respectively. At stage, $ t\geq 2 $, the state of player 1 and 2 are chosen according to the transition probability  $P_{a_{t-1},b_{t-1}}(k_{t-1},k_{t})$ and $Q_{a_{t-1}, b_{t-1}}(l_{t-1},l_{t})$, respectively. Each player does not know the state of the opponent player. Once the initial state is chosen, at stage $ t=1,2...,N$, each player simultaneously chooses its action which is observed by both players.  This is a perfect monitoring and perfect recall game, i.e. every player can observe the current actions of both players, and record all history actions of both players. At stage $ t $, $ G_{k_t,l_t}(a_{t},b_{t}) $ is the one stage payoff of player 1, and the one stage cost of player 2. None of the players can observe the one stage payoff, and the total payoff is revealed to both players at the end of the game.

At the beginning of stage $ t $, the available information of player 1 and 2 is denoted by $\mathcal{I}_t =\lbrace k_1,a_1,b_1,....,k_{t-1},a_{t-1},
b_{t-1},k_{t}\rbrace$ and $\mathcal{J}_t =\lbrace l_1,a_1,b_1,...,l_{t-1},\\ a_{t-1}, b_{t-1},l_{t}\rbrace $, respectively. The behavioral strategy of player 1 and 2 at stage \textit{t} are $\sigma_t$ and $\tau_t$, respectively, where $\sigma_t: \mathcal{I}_t \mapsto \Delta(\mathcal{A})$ and $\tau_t: \mathcal{J}_t \mapsto \Delta(\mathcal{B})$. $\sigma_t$ and $\tau_t$ are the probability distributions over player $1$'s actions $a_t$ and player $2$'s action $b_t$ at stage $t$, respectively. $\Sigma$ and $ \mathcal{T}$ are the sets of behavior strategies of player 1 and 2, respectively. Strategy of player 1, $ \sigma\in\Sigma$, is a sequence of $ \sigma_{t}$ and strategy of player 2, $\tau\in \mathcal{T}$, is a sequence of $\tau_{t}$. The payoff with initial probabilities $ p, q $ and strategies $ \sigma,\tau $ of an \textit{N}-stage $\lambda$ discounted game with $\lambda\in (0,1]$ is defined as
\begin{align*}
\gamma_{N,\lambda}(p,q,\sigma,\tau)=& E_{p,q,\sigma,\tau} \Big(\sum\limits_{t=1}^{N} \lambda^{t-1}G_{k_{t},l_{t}}(a_{t},b_{t})\Big)
\end{align*}
Here, $\lambda\in (0,1)$ if $N$ is infinite.

An $N$-stage $ \lambda $ discounted game $ \Gamma_{N,\lambda}(p,q) $ is defined as two player zero-sum stochastic bayesian game equipped with initial probability $ p $ and $ q $, strategy spaces $ \Sigma $  and $ \mathcal{T} $ and payoff function $ \gamma_{N,\lambda}(p,q,\sigma,\tau)$. If $\lambda =1$ and $N$ is finite then it represents $N$-stage game payoff. If $\lambda < 1$ and $N$ is infinite then it is a discounted game. If $\lambda < 1$ and $N$ is finite then it is a truncated discounted game. Here, we exclude the case when $N=\infty$ and $\lambda=1$. In $ \Gamma_{N,\lambda}(p,q) $, player $1$ wants to maximize the payoff and player $2$ wants to minimize it. Therefore, player $1$ has a security level $\underline{v}_{N}(p,q)$, which is also called the maxmin value of the game, defined as
\begin{align*}
\underline{v}_{N,\lambda}(p,q)= \max\limits_{\sigma\in\Sigma} \min\limits_{\tau\in \mathcal{T}} \gamma_{N,\lambda}(p,q,\sigma,\tau)
\end{align*}
A strategy $ \sigma^{\ast}\in \Sigma $ which can ensure player $1$'s security level, i.e. $ \underline{v}_{N,\lambda}(p,q)= \min\limits_{\tau\in \mathcal{T}} \gamma_{N,\lambda}(p,q,\sigma^{\ast},\tau) $, is called the security strategy of player $1$.

Similarly, player $2$’s security level $\overline{v}_{N,\lambda}(p,q)$, which is also called the minmax value of the game, is defined as
\begin{align*}
\overline{v}_{N,\lambda}(p,q)= \min\limits_{\tau\in \mathcal{T}} \max\limits_{\sigma\in\Sigma}  \gamma_{N,\lambda}(p,q,\sigma,\tau)
\end{align*}
The security strategy of player $2$, $ \tau^{\ast}\in \mathcal{T} $ guarantees player $2$'s security level.

When $\underline{v}_{N,\lambda}(p,q)= \overline{v}_{N,\lambda}(p,q)$, we can say the game has a value denoted by $ v_{N,\lambda}(p,q)= \underline{v}_{N,\lambda}(p,q)= \overline{v}_{N,\lambda}(p,q) $. According to Theorem 3.2 in \cite{mertens1994repeated}, this game has a value if $N$ is finite or $\lambda\in (0,1)$.

\section{Preliminary results}
\label{preliminary results}
\subsection{Primal games and the related results}
\label{sec: primal game}
The N-stage lambda-discounted game is also called the primal game in this paper. It is well known that the primal game has a recursive formula to compute the security strategies of both players, and the believes which will be introduced later, are both players' sufficient statistics. The proof methods in the preliminary results are all typical. We include the proofs in the supplementary materials for reader's convenience.

In the primal game, there are two essential variables $p_t(k_t)=Pr(k_t|\mathcal{I}_t)$ and $q_t(l_t)=Pr(l_t|\mathcal{J}_t)$, which are called the believes of $k_t$ and $l_t$, respectively. With initial condition $p_1=p$ and $q_1=q$, the believes are updated as follows.
\begin{align}
p_{t+1}(k')=&\sum\limits_{k} P_{a,b}(k,k')\frac{p_t(k)X(a,k)}{\bar{x}_{p_t,X}(a)}\doteq p^+_{p_t,X}(k') \label{eq: p+}\\
q_{t+1}(l')=&\sum\limits_{l} Q_{a,b}(l,l')\frac{q_t(l)Y(b,l)}{\bar{y}_{q_t,Y}(b)}\doteq q^+_{q_t,Y}(l') \label{eq: q+}
\end{align}
where $ X(:,k) \in \Delta\mathcal{(A)} $ and $ Y(:,l) \in \Delta\mathcal{(B)} $ are player 1 and 2's strategies at stage $t$ given state $ k$ and $ l $, respectively, $\bar{x}_{p_t,X}(a)= \sum\limits_{k} p_t(k)X(a,k)$, and $\bar{y}_{q_t,Y}(b)= \sum\limits_{l} q_t(l)Y(b,l)$. Notice that the belief $p$ depends on player $1$'s strategy and $q$ depends on player $2$'s strategy.
For the simplicity of the recursive formula that is given in Lemma \ref{lemma: primal, recursive}, we define
\begin{align*}
\bar{G}(p,q,X,Y)= & \sum\limits_{k,l} p(k)q(l)\sum\limits_{a,b}X(a,k)G_{k,l}(a,b)Y(b,l) \\ 
\mathfrak{T}_{p,q,X,Y}(v)=&  \bar{G}(p,q,X,Y)+\lambda \sum_{a\in \mathcal{A}}\sum_{b\in \mathcal{B}} \bar{x}_{p,X}(a) \bar{y}_{q,Y}(b) v(p^{+}_{p,X},q^{+}_{q,Y})
\end{align*}
\begin{lemma}
\label{lemma: primal, recursive}
Consider a primal game $ \Gamma_{N,\lambda}(p,q) $. Let $p_t,q_t$ be the believes at stage $t=1,\ldots,N$ with $n=N+1-t$ stages to go. The game value of the primal game satisfies the following recursive formula
\begin{align}
v_{n,\lambda}(p_t,q_t)=& \max\limits_{X\in\Delta \mathcal{(A)}^\mathcal{K}} \min\limits_{Y\in \Delta (\mathcal{B})^\mathcal{L}} \mathfrak{T}_{p_t,q_t,X,Y}(v_{n-1,\lambda}) \label{eq: primal, recursive, player 1}\\
=&\min\limits_{Y\in \Delta (\mathcal{B})^\mathcal{L}} \max\limits_{X\in\Delta (\mathcal{A})^\mathcal{K}} \mathfrak{T}_{p_t,q_t,X,Y}(v_{n-1,\lambda}) \label{eqn 2}
\end{align}
The optimal solution $X^*$ to equation (\ref{eq: primal, recursive, player 1}) and the optimal solution $Y^*$ to equation (\ref{eqn 2}) are player 1 and 2's security strategies at stage $t$, respectively. Moreover, player 1 and 2's security strategies only depend on $p_t$ and $q_t$, together with $t$ if $N$ is finite. Thus $p_t$ and $q_t$ are the sufficient statistics of the players.
\end{lemma}
The problem of this sufficient statistic $(p_t,q_t)$ is that the belief pair depends on both players' strategy. This requirement is hard to satisfy in zero sum games where players' objectives are against each other. For example, many security problems can be treated as zero-sum games \cite{aziz2020resilience,7417412}, and in these games defenders and attackers will not share their strategies. In this case, no player has full access to the belief pair $(p_t,q_t)$.
\subsection{Dual game and the related results }
\label{sec: dual game}
While the sufficient statistic of primal game is not fully accessible, it is known that the dual game of the primal game has fully accessible sufficient statistic, and with some specially designed initial parameters, the security strategies of the primal game can be recovered by the security strategies in the dual games. These results open a door for strategies based on fully accessible sufficient statistics.

The dual games are rooted from the Fenchel's conjugate of the primal game's game value \cite{sorin2002first,de1999cav}. The Fenchel's conjugate of  $v_{N,\lambda}(p,q)$ regarding $p$ is the game value of type $1$ dual game, and the Fenchel's conjugate of $v_{N,\lambda}(p,q)$ regarding $q$ is the game value of type $2$ dual game.

Type $1$ dual game can be specified by the nine-tuple $ \mathcal{(K,L,A,B},\mu,q,P,Q,G) $, where  $\mathcal{K,L,A},$ $\mathcal{B},q,P,Q,G $ are defined the same as in the primal game and $\mu\in {\Bbb R}^{|\mathcal{K}|} $ is the initial vector payoff over player $1$'s state.
Type $1$ dual game $ \Tilde{\Gamma}_{N,\lambda}^{1}(\mu,q) $ is played similarly as in the primal game $ \Gamma_{N,\lambda}(p,q) $ except that the initial state of player $1$ is chosen by itself rather than the nature. If $ p $ is player $1$'s strategy to choose its initial state, the payoff is
\begin{align*}
\Tilde{\gamma}^{1}_{N,\lambda}(\mu,q,p,\sigma,\tau)= E_{p,q,\sigma,\tau}\Big(\mu(k_1)&+\sum\limits_{t=1}^{N}\lambda^{t-1}
 G_{k_t,l_t}(a_t,b_t)\Big)
\end{align*}

Similarly, type 2 dual game is specified by the nine-tuple $ \mathcal{(K,L,A,B},p,\nu,P,Q,\\
G) $, where  $ \mathcal{K,L,A,}\mathcal{B}, p,P,Q,G $ are defined the same as in the primal game and $\nu\in {\Bbb R}^{|\mathcal{L}|} $ is the initial vector payoff over player 2's state. Type 2 dual game $ \Tilde{\Gamma}_{N,\lambda}^{2}(p,\nu) $ is played similarly as in the primal game $ \Gamma_{N,\lambda}(p,q) $ except that the initial state of player 2 is chosen by itself rather than the nature. If $ q $ is player 2's strategy to choose its initial state, then the payoff function is
\begin{align*}
\Tilde{\gamma}^{2}_{N,\lambda}(p,\nu,q,\sigma,\tau)= E_{p,q,\sigma,\tau}\Big(\nu(l_1)&+\sum\limits_{t=1}^{N}\lambda^{t-1} G_{k_t,l_t}(a_t,b_t)\Big)
\end{align*}
In both dual games, player 1 wants to maximize the payoff and player 2 wants to minimize it. The game values of type 1 and 2 dual games are denoted by $w_{N,\lambda}^1(\mu,q)$ and $w^{2}_{N,\lambda}(p,\nu)$, respectively. The game values of the primal and dual games have the following relationship.

\begin{proposition}
\label{proposition 3.1}
Let $v_{n,\lambda}(p,q)$ be the game value of the primal game $\Gamma_{n,\lambda}(p,q)$, and $w_{n,\lambda}^1(\mu,q)$ and $w^{2}_{n,\lambda}(p,\nu)$ be the game values of type 1 and 2 dual games, respectively. We have
\begin{align}
v_{n,\lambda}(p,q)=&\min_{\mu} w^{1}_{n,\lambda}(\mu,q)-p^{T}\mu \label{eqn 3}\\
v_{n,\lambda}(p,q)=&\max_{\nu} w^{2}_{n,\lambda}(p,\nu)-q^{T}\nu \label{eqn 4}\\
w^{1}_{n,\lambda}(\mu,q)=& \max_{p} p^{T}\mu +v_{n,\lambda}(p,q) \label{eqn 5}\\
w^{2}_{n,\lambda}(p,\nu)= &\min_{q} q^{T}\nu +v_{n,\lambda}(p,q) \label{eqn 6}
\end{align}
Moreover, player 2's optimal strategy in type 1 dual game $\Tilde{\Gamma}^{1}_{n,\lambda}(\mu^{*},q)$ is also its optimal strategy in primal game $\Gamma_{n,\lambda}(p,q)$, where $\mu^*$ is the optimal solution to equation (\ref{eqn 3}) and player 1's optimal strategy in type 2 dual game $\Tilde{\Gamma}^{2}_{n,\lambda}(p,\nu^{*})$ is also its optimal strategy in primal game $\Gamma_{N,\lambda}(p,q)$, where $\nu^*$ is the optimal solution to equation (\ref{eqn 4}).
\end{proposition}
The proof of this proposition is in the supplementary material. 
\begin{corollary}
\label{corollary: 3.2}
The optimal solutions to equation (\ref{eqn 3}) and equation (\ref{eqn 4}) are
\begin{align}
\mu^*=&-\beta(\tau^*) \label{eq: mu star} \\
\nu^*=&-\alpha(\sigma^*) \label{eq: nu star}
\end{align}
respectively. Here, $\sigma^*$ and $\tau^*$ are player 1 and 2's security strategies in the primal game. $\alpha(\sigma)= (\alpha_l(\sigma))_{l\in L}$ and $\beta(\tau)=(\beta_k(\tau))_{k\in K}$ are defined as follows. 
\begin{align*}
\alpha_{l}(\sigma) =& \min_{\tau} E_{p,q,\sigma,\tau}\Big( \sum_{t=1}^{N} \lambda^{t-1}G_{k_t,l_t}(a_t,b_t)| l_1=l \Big)\\
\beta_{k}(\tau) =& \max_{\sigma} E_{p,q,\sigma,\tau}\Big( \sum_{t=1}^{N} \lambda^{t-1}G_{k_t,l_t}(a_t,b_t)| k_1=k \Big)
\end{align*}
Moreover, $w^{1}_{n,\lambda}(\mu^*,q)=w^{2}_{n,\lambda}(p,\nu^*)=0$, where $\mu^*$ and $\nu^*$ are given in equation (\ref{eq: mu star}) and (\ref{eq: nu star}).
\end{corollary}
Based on the recursive formula of primal game and the relationship between the primal game and the two dual games, the recursive formula of the dual games can be derived as follows.
\begin{proposition}
\label{prop: recursive formula, dual games}
The value of type 1 dual game satisfies the following recursive formula:
\begin{align}
w^{1}_{n,\lambda}(\mu,q)=& \min_{Y} \min_{\beta_{a\in A, b\in B}} \max_{\Pi} \sum_{a,k} \Pi(a,k) \Big[ \mu(k) + \sum_{l, b} G_{k,l} (a,b)Y(b,l)q(l)+ \nonumber \\
& \lambda \sum_{b} \bar{y}_{q,Y}(b) \Big( w^{1}_{n-1,\lambda}(\beta_{a,b}, q^+_{q,Y}) - \sum_{k'}P_{a,b}(k,k')\beta_{a,b}(k') \Big)\Big] \label{dual_recursive_type_1} \\
\doteq & \min_{Y} \min_{\beta} \max_{\Pi} \mathfrak{S}^1_{\mu,q,Y,\beta,\Pi}(w^1_{n-1,\lambda}) \nonumber
\end{align}
where $\Pi(a,k)=Pr(a \cap k)$. The value of type 2 dual game satisfies the following recursive formula:
\begin{align}
w^{2}_{n,\lambda}(p,\nu)=& \max_{X} \max_{\alpha_{a\in A, b\in B}} \min_{\Psi} \sum_{b,l} \Psi(b,l) \Big[ \nu(l) + \sum_{k,  a} G_{k,l}(a,b) X(b,l) p(k)+ \nonumber \\
&  \lambda \sum_{a} \overline{x}_{p,X}(a)\Big( w^{2}_{n-1,\lambda}(p^{+}_{p,X}, \alpha_{a,b})- \sum_{l'} Q_{a,b}(l,l')\alpha_{a,b}(l')\Big)\Big] \label{dual_recursive_type_2} \\
\doteq & \max_{X} \max_{\alpha} \min_{\Psi} \mathfrak{S}^2_{p,\nu,X,\alpha,\Psi}(w^2_{n-1,\lambda}) \nonumber
\end{align}
where $\Psi(b,l)=Pr(b \cap l)$.
\end{proposition}
Based on Proposition \ref{prop: recursive formula, dual games}, we see that besides the belief, there is another important variable, $\mu$ in type 1 dual game, and $\nu$ in type 2 dual game. We call them the vector payoffs on state $k$ and $l$, respectively, and define them in the following way. 
\begin{definition}
\label{def: vector payoff}
Consider a type 1 dual game $\Tilde{\Gamma}^{1}_{N,\lambda}(\mu,q)$. Let $\mu_1=\mu$ be the initial vector payoff, and $\beta^*$ be the optimal solution to the following problem
\begin{align}
w^{1}_{n,\lambda}(\mu_t,q_t)=\min_{Y} \min_{\beta_{a \in A, b \in B}} \max_{\Pi} \mathfrak{S}^1_{\mu_t,q_t,Y,\beta,\Pi}(w^1_{n-1,\lambda}) \label{eq: optimal problem, dual 1}
\end{align}
where $n=N+1-t$ for stage $t$. The vector payoff over state $k$ at stage $t+1$ is defined as $\mu_{t+1}=\beta_{a,b}^*$. Similarly, consider a type $2$ dual game $\Tilde{\Gamma}^{2}_{N,\lambda}(p,\nu)$. Let $\nu_1=\nu$ be the initial vector payoff, and $\alpha^*$ be the optimal solution to the problem
\begin{align}
w^{2}_{n,\lambda}(p_t,\nu_t)=\max_{X} \max_{\alpha_{a \in A, b \in B}} \min_{\Psi} \mathfrak{S}^2_{p_t,\nu_t,X,\alpha,\Psi}(w^2_{n-1,\lambda})\label{eq: optimal problem, dual 2}
\end{align}
where $n=N+1-t$ for stage $t$. The vector payoff over state $l$ at stage $t+1$ is defined as $\nu_{t+1}=\alpha_{a,b}^*$.
\end{definition}
From Proposition \ref{prop: recursive formula, dual games} we can see, $(p_t,\nu_t)$ and $(\mu_t,q_t)$ are the sufficient statistics in type 2 and 1 dual games.
\begin{corollary}
\label{col: sufficient statistics, dual games}
Player $1$'s security strategy at stage $t$ in type $2$ dual game $\Tilde{\Gamma}^{2}_{N,\lambda}(p,\nu)$ only depends on $p_t$ and $\nu_t$, together with $t$ if $N$ is finite.
Player $2$'s security strategy at stage $t$ in type $1$ dual game $\Tilde{\Gamma}^{1}_{N,\lambda}(\mu,q)$ only depends on $\mu_t$ and $q_t$, together with $t$ if $N$ is finite.
\end{corollary}
Equation (\ref{eq: optimal problem, dual 1}) and (\ref{eq: q+}) imply that $(\mu_{t+1},q_{t+1})$ only depends on player 2's strategy and state information and hence is fully accessible to player $2$. Equation (\ref{eq: optimal problem, dual 2}) and (\ref{eq: p+}) imply that $(p_{t+1}, \nu_{t+1})$ only depends on player 1's strategy and state information and hence is fully accessible to player $1$.  
\section{Problem Statement}
\label{Problem Statement}
The literature has provided rich descriptive results about the existence and properties of the game values and the security strategies in both primal and dual games \cite{2226507620060801, gensbittel2015value, Rosenberg1998, von1996efficient, sorin2003stochastic}, and hence builds a solid foundation for us to explore prescriptive results about how to efficiently store and compute the strategies. As pioneer work in computing security strategies in Bayesian games, \cite{von1996efficient} provides LPs to solve two-player zero-sum Bayesian games. With the sequence form, \cite{von1996efficient} reduces the computational complexity to be linear with respect to the number of the leaf nodes in a game tree. Although this is a great decrease in computational complexity compared to the extensive form, since the number of the leaf nodes grows exponentially with respect to the time horizon, the size of the LPs also grows accordingly. This fact limits us to games with small time horizons.

This paper focuses on computing strategies in a finite horizon Bayesian game with a large time horizon. While \cite{von1996efficient} provides an LP formulation to compute the security strategies of the primal stochastic Bayesian game, the player has to remember all the history actions of both players to figure out which strategy to use, and, as mentioned before, the size of the LP increases exponentially with respect to the number of stages. We are interested in a strategy which is based on a fully accessible sufficient statistic and has constant computational complexity which does not grow with respect to the full time horizon. So, the player only needs to remember a sufficient statistic and a fixed sized action history and updates its strategy periodically.

Based on the previous results in the literature which suggest that only dual games have fully accessible sufficient statistics and the security strategies in the dual games can serve as the security strategies in the primal game, a promising method is to divide the total game into multiple \emph{windows}, solve it window by window and update the strategy periodically. We refer this method as \emph{window by window method}.

There are several challenges in this solution. First, we don't know the special initial vector payoffs $\mu^*$ and $\nu^*$ of the dual games for which the security strategies will also be the security strategies in the primal game. Second, we don't know how to efficiently compute the sufficient statistics in the dual games for the next window, though we know they shall be the optimal solution in equation (\ref{eq: optimal problem, dual 1}) and (\ref{eq: optimal problem, dual 2}). Third, as the window-by-window method is a suboptimal strategy, we need to analyze the performance of this algorithm and check how close it is to the optimal game value.
\section{Sufficient statistic based strategies and the LP formulations for short time horizon}
\label{LPs for sufficient statistics based strategies}
Assume the computational capacity only allow us to compute security strategies in games with $n<N$ stages in a timely manner. For every $n$ stages, we can compute and apply the security strategies in the $n$-stage game periodically. The problem is that, except the first window, no player has full access to $(p_t,q_t)$ in the other windows and cannot compute the security strategy in the primal game.

To solve this problem, we propose a window-by-window dual game based strategy as follows. 
\begin{enumerate}
  \item The first $n$-stage window.
    \begin{enumerate}
      \item At the beginning, compute the initial vector payoffs in both $n$-stage dual games defined in equation (\ref{eq: mu star}) and (\ref{eq: nu star}) and the security strategies in the $n$-stage primal game.
      \item In stage $t=1,\ldots,n$ inside this window, apply the security strategy at stage $t$ and update the believes (equation (\ref{eq: p+}-\ref{eq: q+})) and the vector payoffs (Definition \ref{def: vector payoff}).
    \end{enumerate} 
  \item The $2^{nd}$, $3^{rd}$, $\ldots$, $\lfloor \frac{N}{n} \rfloor^{th}$ $n$-stage window.
    \begin{enumerate}
      \item At the beginning, compute the security strategies in the $n$-stage dual games based on the updated believes and vector payoffs.
      \item In stage $t=1,\ldots,n$ inside the window, apply the security strategy at stage $t$ and update the believes (equation (\ref{eq: p+}-\ref{eq: q+})) and the vector payoffs (Definition \ref{def: vector payoff}).
    \end{enumerate}
  \item If $N\ (mod\ n) \neq 0$ then the last window size $m=N-(\lfloor \frac{N}{n} \rfloor) n$ stages.
    \begin{enumerate}
      \item At the beginning, compute the security strategies in the $m$-stage dual games based on the updated believes and vector payoffs.
      \item In stage $t=1,\ldots,m$ inside the window, apply the security strategy at stage $t$.
    \end{enumerate}
\end{enumerate}
The rest of this section will provide LPs to compute the initial vector payoffs, the updated vector payoffs, and the security strategies. To conclude this section, a detailed algorithm is provided.
\subsection{LPs to compute the initial vector payoffs in dual games}
First of all, we would like to introduce two realization plans and two weighted vector payoffs that will be used to compute the initial vector payoffs in dual games. The idea is introduced in \cite{von1996efficient} and the realization plan of player 1 is a function $\mathcal{R_{\mathcal{I}}}: \sigma \rightarrow \mathbb{R}$. The player 1's realization plan $R_{\mathcal{I}_t}(a_t)$ and player 2's realization plan $S_{\mathcal{J}_t}(b_t)$ are defined as follows.
\begin{align*}
R_{\mathcal{I}_t}(a_t)=&p(k_1)\prod_{s=1}^{t-1} P_{a_s,b_s}(k_s,k_{s+1}) \prod_{s=1}^{t} \sigma_s^{a_s}(\mathcal{I}_s)\\
S_{\mathcal{J}_t}(b_t)=&q(l_1)\prod_{s=1}^{t-1} Q_{a_s,b_s}(l_s,l_{s+1}) \prod_{s=1}^{t} \tau_s^{b_s}(\mathcal{J}_s)
\end{align*}
with $R_{\mathcal{I}_0}(a_0)=p(k_1)$, $P_{a_0,b_0}(k_0,k_1)=1$, $S_{\mathcal{J}_0}(b_0)=q(l_1)$ and $Q_{a_0,b_0}(l_0,l_1)=1$. It is straight forward to show
\begin{align}
R_{\mathcal{I}_t}(a_t)=& P_{a_{t-1},b_{t-1}}(k_{t-1},k_t) \sigma_{t}^{a_t}(\mathcal{I}_t)R_{\mathcal{I}_{t-1}}(a_{t-1}) \label{eqn 7}\\
S_{\mathcal{J}_t}(b_t)=& Q_{a_{t-1},b_{t-1}}(l_{t-1},l_t) \tau_{t}^{b_t}(\mathcal{J}_t)S_{\mathcal{J}_{t-1}}(b_{t-1}) \label{eqn 8}
\end{align}
The weighted payoff $U_{\mathcal{J}_t}(\sigma,\tau)$ and $Z_{\mathcal{I}_t}(\sigma,\tau)$ of player 1 and 2 are defined as follows.
\begin{align*}
U_{\mathcal{J}_{t}}(\sigma,\tau)=& \sum_{k_1,...k_t} R_{\mathcal{I}_{t-1}}(a_{t-1})P_{a_{t-1},b_{t-1}}(k_{t-1},k_{t}) E\Big(\sum_{s=t}^{N}
\lambda^{s-1} G_{k_s,l_s}(a_s,b_s)|k_1,...k_t,\mathcal{J}_t \Big)\\
Z_{\mathcal{I}_{t}}(\sigma,\tau)=& \sum_{l_1,...l_t} S_{\mathcal{J}_{t-1}}(b_{t-1})Q_{a_{t-1},b_{t-1}}(l_{t-1},l_{t}) E\Big(\sum_{s=t}^{N}
\lambda^{s-1} G_{k_s,l_s}(a_s,b_s)|l_1,...l_t,\mathcal{I}_t \Big)
\end{align*}
Based on the definition of $U$ and $Z$, we see that $\alpha_l(\sigma^*)$ and $\beta_k(\tau^*)$ are related to $U$ and $Z$ in the following way.
\begin{align}
\alpha_{l}(\sigma^{\ast})=& \min_{\tau} U_{\mathcal{J}_1}(\sigma^{\ast},\tau); \hbox{where $\mathcal{J}_1=\{l\}$} \label{eq: alpha star}\\
\beta_{k}(\tau^{\ast})=& \max_{\sigma} Z_{\mathcal{I}_1}(\sigma,\tau^{\ast}); \hbox{where $\mathcal{I}_1=\{k\}$} \label{eq: beta star}
\end{align}
To build the LP formulation to compute $\alpha_{l}(\sigma^{\ast})$ and $\beta_{k}(\tau^{\ast})$ and hence $\mu^*$ and $\nu^*$ as in Corollary \ref{corollary: 3.2}, we first compute recursive formulas for $\min_{\tau} U_{\mathcal{J}_t}(\sigma,\tau)$ and  $\max_{\sigma} Z_{\mathcal{I}_t}(\sigma,\tau)$ in Corollary \ref{col: recursive formula, U star, Z star} based on the recursive formulas for $ U_{\mathcal{J}_t}(\sigma,\tau) $ and $Z_{\mathcal{I}_t}(\sigma,\tau)$ in lemma \ref{Lemma 4.2} which is in the Appendix. Then set up LP's for $\min_{\tau} U_{\mathcal{J}_t}(\sigma,\tau)$ and $\max_{\sigma} Z_{\mathcal{I}_t}(\sigma,\tau)$ in Lemma \ref{lemma: LP, U star, Z star} and finally compute $\alpha_{l}(\sigma^{\ast})$ and $\beta_{k}(\tau^{\ast})$ in Theorem \ref{lp of primal game}.
Define player $1$ and $2$'s optimal weighted payoffs $U^{*}_{\mathcal{J}_t}(\sigma)$ and $Z^{*}_{\mathcal{I}_t}(\tau)$ as
\begin{align*}
U^{*}_{\mathcal{J}_t}(\sigma) =& \min_{\tau_{t:N}} U_{\mathcal{J}_t}(\sigma,\tau)\\
Z^{*}_{\mathcal{I}_t}( \tau) =& \max_{\sigma_{t:N}} Z_{\mathcal{I}_t}(\sigma,\tau)
\end{align*}
where $\tau_{t:N}$ and $\sigma_{t:N}$ is player 2 and 1's behavior strategy from stage $t$ to $N$, respectively. It is straight forward to show the following corollary based on Lemma \ref{Lemma 4.2} in Appendix.
\begin{corollary}
\label{col: recursive formula, U star, Z star}
The optimal weighted payoffs $ U^{*}_{\mathcal{J}_t}(\sigma) $ and  $ Z^{*}_{\mathcal{I}_t}(\tau) $ satisfy the following recursive formulas.
\begin{align*}
U^{*}_{\mathcal{J}_t}(\sigma) =& \min_{\tau_{t}} \sum_{b_t} \Big(\sum_{a_t} \sum_{k_1,...k_t} R_{\mathcal{I}_t}(a_t)\lambda^{t-1}
 G_{k_t,l_t}(a_t,b_t) + \sum_{a_t} \sum_{l_{t+1}} Q_{a_t,b_t}(l_t,l_{t+1})\\ & U_{\mathcal{J}_{t+1}}^{*}(\sigma)\Big) \tau_{t}^{b_t}(\mathcal{J}_t); \ \ with\ U^{*}_{\mathcal{J}_{N+1}}(\sigma)=0\\
Z^{*}_{\mathcal{I}_t}(\tau) =& \max_{\sigma_{t}} \sum_{a_t} \Big(\sum_{b_t} \sum_{l_1,...l_t} S_{\mathcal{J}_t}(b_t)\lambda^{t-1} G_{k_t,l_t}(a_t,b_t)
+ \sum_{b_t} \sum_{k_{t+1}} P_{a_t,b_t}(k_t,k_{t+1})\\
& Z_{\mathcal{I}_{t+1}}^{*}(\tau)\Big) \sigma_{t}^{a_t}(\mathcal{I}_t);\ \ with Z^{*}_{\mathcal{I}_{N+1}}(\tau)=0
\end{align*}
\end{corollary}
The optimal weighted payoff $ U^{*}_{\mathcal{J}_t}(\sigma) $ and  $ Z^{*}_{\mathcal{I}_t}(\tau) $ can be computed by LP given in Lemma \ref{lemma: LP, U star, Z star}.
\begin{lemma}
\label{lemma: LP, U star, Z star}
For any $t=1,\ldots,N$, $U_{\mathcal{J}_t}^{\ast}(\sigma)$ and $ Z^{*}_{\mathcal{I}_t}(\tau) $ can be computed by the following LPs.
\begin{align*}
U_{\mathcal{J}_t}^{\ast}(\sigma) &= \max_{U} U_{\mathcal{J}_t}\\
s.t.
& \sum_{k_1...k_s} \sum_{a_s} R_{\mathcal{I}_s}(a_s)\lambda^{s-1} G_{k_s,l_s}(a_s,b_s) + \sum_{l_{s+1}} \sum_{a_s}Q_{a_s,b_s}
(l_s,l_{s+1})U_{\mathcal{J}_{s+1}} \geq U_{\mathcal{J}_s};\\
& \forall s=t,...N; \forall b_s; \forall \mathcal{J}_t \subset \mathcal{J}_s; U_{\mathcal{J}_{N+1}}=0, \forall \mathcal{J}_{N+1}\\
Z_{\mathcal{I}_t}^{*}(\tau) &= \min_{Z} Z_{\mathcal{I}_t}\\
s.t.
& \sum_{l_1...l_s} \sum_{b_s} S_{\mathcal{J}_s}(b_s)\lambda^{s-1} G_{k_s,l_s}(a_s,b_s) + \sum_{k_{s+1}} \sum_{b_s}P_{a_s,b_s}(k_s,k_{s+1}) Z_{\mathcal{I}_{s+1}} \leq Z_{\mathcal{I}_s};\\
& \forall s=t...N; \forall a_s; \forall \mathcal{I}_t \subset \mathcal{I}_s; Z_{\mathcal{I}_{N+1}}=0, \forall \mathcal{I}_{N+1}
\end{align*}
\end{lemma}
Based on the LP in Lemma \ref{lemma: LP, U star, Z star}, we further develop the LP's to compute $U_{\mathcal{J}_t}^{\ast}(\sigma^*)$ and $Z_{\mathcal{I}_t}^{*}(\tau^*)$ by considering the game value and the security strategies of the primal game in Theorem \ref{lp of primal game}.
\begin{theorem}
\label{lp of primal game}
Consider a primal game $\Gamma_{n,\lambda}(p,q)$. Its game value satisfies
\begin{align}
v_{n,\lambda}(p,q)&= \max_{R} \max_{U} \sum_{l}q(l)  U_{\mathcal{J}_1}; \quad where\ \mathcal{J}_{1}=\{l\} \label{eqn new 13}\\
s.t.\ &\sum_{k_1...k_t} \sum_{a_t} R_{\mathcal{I}_t}(a_t) \lambda^{t-1} G_{k_t,l_t}(a_t,b_t)+ \sum_{l_{t+1}} \sum_{a_t}Q_{a_t,b_t}(l_t,l_{t+1})U_{\mathcal{J}_{t+1}} \geq U_{\mathcal{J}_t}; \nonumber  \\
& \forall t=1,...n, \forall b_t, \forall \mathcal{J}_{t+1}\supset \mathcal{J}_t \label{eq: constraint 1, primal}\\
&\sum_{a_t} R_{\mathcal{I}_t}(a_t) = P_{a_{t-1},b_{t-1}}(k_{t-1},k_t)R_{\mathcal{I}_{t-1}}(a_{t-1});\forall t=1,...n, \forall \mathcal{I}_{t} \label{eq: constraint 2, primal}\\
&  U_{\mathcal{J}_{n+1}}=0; \quad \forall \mathcal{J}_{n+1}\label{eq: constraint 3, primal}\\
& R_{\mathcal{I}_t}(a_t) \geq 0;   \quad \forall \mathcal{I}_{t} \label{eq: constraint 4, primal}
\end{align}
The optimal strategy of player 1 is
\begin{align}
\sigma_{t}^{*a_t}(\mathcal{I}_t) =&\frac{ R^*_{\mathcal{I}_{t}}(a_{t})}{P_{a_{t-1},b_{t-1}}(k_{t-1},k_t)R^*_{\mathcal{I}_{t-1}}(a_{t-1})} \label{eqn new 14}
\end{align}
and the initial vector payoff of type 2 dual game defined in equation (\ref{eq: nu star}) is 
\begin{align}
\nu^*_l=&-\alpha_{l}(\sigma^{\ast}) =-U_{\mathcal{J}_1}^{\ast} \label{eqn new 15}
\end{align}
where $R^*$ and $U^*$ are the optimal solution to (\ref{eqn new 13}). Similarly, we have
\begin{align}
\begin{split}
v_{n,\lambda}(p,q)=& \min_{S} \min_{Z} \sum_{k}p(k)  Z_{\mathcal{I}_1};\ \ where\ \mathcal{I}_{1}=\{k\}\label{eqn 13}\\
s.t.
&\sum_{l_1...l_t} \sum_{b_t} S_{\mathcal{J}_t}(b_t)\lambda^{t-1} G_{k_t,l_t}(a_t,b_t)+ \sum_{k_{t+1}} \sum_{b_t} P_{a_t,b_t}(k_t,k_{t+1})Z_{\mathcal{I}_{t+1}} \leq Z_{\mathcal{I}_t};\\
& \forall t=1,...n\ ,\forall a_t, \forall\mathcal{I}_{t+1}\supset \mathcal{I}_t \\
&\sum_{b_t} S_{\mathcal{J}_t}(b_t) = Q_{a_{t-1},b_{t-1}}(l_{t-1},l_t) S_{\mathcal{J}_{t-1}}(b_{t-1}); \forall t=1:n, \forall \mathcal{J}_t \\
&  Z_{\mathcal{I}_{n+1}}=0; \quad \forall \mathcal{I}_{n+1} \\
& S_{\mathcal{J}_t}(b_t) \geq 0;\quad \forall \mathcal{J}_{t}
\end{split}
\end{align}
The optimal strategy of player 2 is
\begin{align}
\tau_{t}^{*b_t}(\mathcal{J}_t)=& \frac{S^*_{\mathcal{J}_{t}}(b_{t})}{Q_{a_{t-1},b_{t-1}}(l_{t-1},l_t)S^*_{\mathcal{J}_{t-1}}(b_{t-1})} \label{eqn new 16}
\end{align}
and the initial vector payoff of type 1 dual game defined in equation (\ref{eq: mu star}) is
\begin{align}
\mu^*_k=-\beta_{k}(\tau^{\ast}) =- Z_{\mathcal{I}_1}^{\ast} \label{eqn new 18}
\end{align}
where $Z^*$ and $S^*$ are the optimal solution to (\ref{eqn 13}).
\end{theorem}
\begin{proof}
According to the discussion in the second paragraph on page 248 in \cite{von1996efficient}, the realization probabilities can serve as strategic variables of a player. We have,
\begin{align*}
v_{n,\lambda}(p,q)=& \max_{\sigma} \min_{\tau} E\Big( \sum_{t=1}^{n}\lambda^{t-1} G_{k_t,l_t}(a_t,b_t)\Big)\\
=& \max_{R} \sum_{l} q(l)  \min_{\tau^{l}} E\Big( \sum_{t=1}^{n}\lambda^{t-1} G_{k_t,l_t}(a_t,b_t)|l \Big)\\
=& \max_{R} \sum_{l} q(l) U^{\ast}_{\mathcal{J}_1}(\sigma)\\
=& \max_{R} \max_{U} \sum_{l}q(l)  U_{\mathcal{J}_1},where\ \mathcal{J}_{1}=\{l\}\\
& s.t.\ equation (\ref{eq: constraint 1, primal}-\ref{eq: constraint 4, primal})
\end{align*}
Let $R^*$ and $U^*$ be the optimal solution. Equation (\ref{eqn 7}) implies that the optimal strategy satisfies equation (\ref{eqn new 14}), and equation (\ref{eq: alpha star}) implies equation (\ref{eqn new 15}). Following the same steps, LP (\ref{eqn 13}), equation (\ref{eqn new 16}) and (\ref{eqn new 18}) can be proved. \qed
\end{proof}
The size of the linear program (\ref{eqn new 13}-\ref{eq: constraint 4, primal}) is polynomial with respect to the size of $|\mathcal{K}|, |\mathcal{L}|, |\mathcal{A}|, |\mathcal{B}|$ and exponential with respect to the length of the stage $n$. We first examine the size of variables. When $t=1$ the number of scalar variables for $R_{\mathcal{I}_t}$ is $|\mathcal{K}||\mathcal{A}|$. When $t=2$ the number of scalar variables for $R_{\mathcal{I}_t}$ is $|\mathcal{K}||\mathcal{A}||\mathcal{B}||\mathcal{K}||\mathcal{A}|$. So, when $t=n$ the number of scalar variables for $R_{\mathcal{I}_t}$ is $|\mathcal{K}|(|\mathcal{A}||\mathcal{B}||\mathcal{K}|)^{n-1}|\mathcal{A}|$. The number of scalar variables for $U_{\mathcal{J}_t}$ is $|\mathcal{L}|(|\mathcal{A}||\mathcal{B}||\mathcal{L}|)^{n-1}$. Therefore, the number of scalar variables has the order of $O(|A|^n|B|^n(|K|^n+|L|^n))$. For the constraints, constraint set (\ref{eq: constraint 1, primal}) has $|\mathcal{B}||\mathcal{L}|$ inequalities when $t=1$, $|\mathcal{B}||\mathcal{L}||\mathcal{A}||\mathcal{B}||\mathcal{L}|$ inequalities when $t=2$, and $|\mathcal{B}||\mathcal{L}|(|\mathcal{A}||\mathcal{B}||\mathcal{L}|)^{n-1}$ inequalities when $t=n$. In constraint set (\ref{eq: constraint 2, primal}), the number of equalities is $|\mathcal{K}|$ for $t=1$, $|\mathcal{K}||\mathcal{A}||\mathcal{B}||\mathcal{K}|$ for $t=2$, and $|\mathcal{K}|(|\mathcal{A}||\mathcal{B}||\mathcal{K}|)^{n-1}$ for $t=n$. There are $|\mathcal{L}|(|\mathcal{A}||\mathcal{B}||\mathcal{L}|)^{n}$ and $|\mathcal{K}|(|\mathcal{A}||\mathcal{B}||\mathcal{K}|)^{n}|\mathcal{A}|$ equations in constraint set (\ref{eq: constraint 3, primal}) and (\ref{eq: constraint 4, primal}), respectively. Therefore, the number of the constraints is in the order of $O(|A|^n|B|^n(|K|^n+|L|^n))$. The computational complexity of LP (\ref{eqn 13}) is the same as that of LP (\ref{eqn new 13}-\ref{eq: constraint 4, primal}).
\subsection{LPs to compute the security strategies of dual games}
Theorem \ref{lp of primal game} provides LPs to compute the initial vector payoffs and the optimal strategies of the players at the same time. Based on the relation between the game values of the primal game and the two dual games in Proposition \ref{proposition 3.1}, we can further develop LPs to compute the security strategy of player 1 in type 2 dual game, and the security strategy of player 2 in type 1 dual game.
\begin{proposition}
\label{prop: LP, dual game}
The game value of type 2 dual game $\tilde{\Gamma}^2_{n,\lambda}(p,\nu)$ satisfies
\begin{align}
w^{2}_{n,\lambda}(p,\nu)=& \max_{R} \max_{U_\mathcal{J}} \max_{U_0} U_0 \label{eqn new 19}\\
s.t.\ & \nu(l)+ U_{\mathcal{J}_1} \geq U_{0}, \quad \forall l,where \ \mathcal{J}_1=\{l\} \nonumber \\
& \sum_{k_1...k_t} \sum_{a_t} R_{\mathcal{I}_t}(a_t) \lambda^{t-1} G_{k_t,l_t}(a_t,b_t)+ \sum_{l_{t+1}} \sum_{a_t} Q_{a_t,b_t}(l_t,l_{t+1}) \nonumber
\end{align}
\begin{align}
& U_{\mathcal{J}_{t+1}} \geq U_{\mathcal{J}_t}; \quad \forall t=1,...n,\forall b_t,\forall \mathcal{J}_{t+1}\supset \mathcal{J}_t\nonumber\\
& \sum_{a_t} R_{\mathcal{I}_t}(a_t) = P_{a_{t-1},b_{t-1}}(k_{t-1},k_t)R_{\mathcal{I}_{t-1}}(a_{t-1});\quad \forall t=1,...n;  \forall \mathcal{I}_{t} \nonumber\\
& U_{\mathcal{J}_{N+1}}=0; \quad \forall \mathcal{J}_{N+1}\nonumber\\
& R_{\mathcal{I}_t}(a_t) \geq 0;  \quad \forall \mathcal{I}_{t} \nonumber
\end{align}
The optimal strategy of player 1 in type 2 dual game can be computed from equation (\ref{eqn new 14}) where $R^*$ and $U^*$ is the optimal solution to problem (\ref{eqn new 19}). Similarly, the game value of type 1 dual game $\tilde{\Gamma}^1_{n,\lambda}(\mu,q)$ satisfies
\begin{align}
w^{1}_{n,\lambda}(\mu,q)=& \min_{S} \min_{Z_{\mathcal{I}}} \min_{Z_0} Z_0  \label{eqn 14}\\
s.t.\
& \mu(k)+ Z_{\mathcal{I}_1} \leq Z_{0}, \quad \forall k, where \mathcal{I}_1=\{k\} \nonumber \\
&\sum_{l_1...l_t} \sum_{b_t} S_{\mathcal{J}_t}(b_t){\lambda^{t-1}} G_{k_t,l_t}(a_t,b_t)+ \sum_{k_{t+1}} \sum_{b_t}P_{a_t,b_t}(k_t,k_{t+1}) \nonumber \\
&Z_{\mathcal{I}_{t+1}} \leq Z_{\mathcal{I}_t}; \quad \forall t=1,...n; \forall a_{t}, \forall \mathcal{I}_{t+1}\supset \mathcal{I}_t \nonumber\\
&\sum_{b_t} S_{\mathcal{J}_t}(b_t) = Q_{a_{t-1},b_{t-1}}(l_{t-1},l_t) S_{\mathcal{J}_{t-1}}(b_{t-1}); \quad \forall t=1,...n;  \forall \mathcal{J}_{t}\nonumber \\
& Z_{\mathcal{I}_{n+1}}=0; \quad \forall \mathcal{I}_{n+1}\nonumber\\
& S_{\mathcal{J}_t}(b_t) \geq 0;\quad \forall \mathcal{J}_{t} \nonumber
\end{align}
The optimal strategy of player 2 in type 1 dual game can be computed from equation (\ref{eqn new 16}) where $S^*$ and $Z^*$ are the optimal solution to problem (\ref{eqn 14}).
\end{proposition}
\begin{proof}
By using theorem \ref{lp of primal game} we can write equation (\ref{eqn 6}) as,
\begin{align}
w^{2}_{n,\lambda}(p,\nu) =& \min_{q} q^{T}\nu + \max_{R} \max_{U} \sum_{l}q(l)  U_{\mathcal{J}_1} \nonumber \\
=& \min_{q} \max_{R} \max_{U} \sum_{l} q(l)\big(\nu(l)+U_{\mathcal{J}_1}\big) \label{eqn 34} \\
=& \min_{q} \max_{R} \max_{U} F_{1}(q,R,U) \nonumber\\
F_{1}(q,R,U) =& \sum_{l} q(l)\big(\nu(l)+U_{\mathcal{J}_1}\big) \nonumber\\
& s.t. \nonumber\\
& \sum_{k_1...k_t} \sum_{a_t} R_{\mathcal{I}_t}(a_t) G_{k_t,l_t}(a_t,b_t) + \sum_{l_{t+1}} \sum_{a_t} Q_{a_t,b_t}(l_t,l_{t+1}) U_{\mathcal{J}_{t+1}} \geq U_{\mathcal{J}_t}; \nonumber\\
& \forall t=1,...n; \forall b_t; \forall \mathcal{J}_{t+1} \subset \mathcal{J}_{t} \label{eqn 16} \\
& \sum_{a_t}R_{\mathcal{I}_t}(a_t) = P_{a_{t-1},b_{t-1}}(k_{t-1},k_{t}) R_{\mathcal{I}_{t-1}}(a_{t-1}); \forall t=1,...n; \forall \mathcal{I}_t \label{eqn 17} \\
& R_{\mathcal{I}_t}(a_t)\geq 0 \label{eqn 18} \\
&\sum_{l} q(l)=1  \label{eqn 19} \\
& q(l) \geq 0; \label{eqn 20} \quad \forall l
\end{align}
Given $U$, $R$ s.t. equation (\ref{eqn 34})-(\ref{eqn 18}) are linear constraints. Hence $R$ is in a compact convex set. Given $U$, $q$ s.t. equation (\ref{eqn 19}) and (\ref{eqn 20}) are linear constraints. Hence $q$ is in a compact convex set. Given $U$, $F_{1}(q,R,U)$ is linear w.r.t. $q$ and $R$. According to the minimax theorem we can write,
\begin{align*}
\max_{R} \min_{q} \max_{U} & \sum_{l} q(l)\big(\nu(l)+U_{\mathcal{J}_1}\big)\\
& s.t.\ equation\ (\ref{eqn 16})\ to\ (\ref{eqn 20})
\end{align*}
Given $R$, $U$ s.t. equation (\ref{eqn 34}), which is linear. Hence $U$ is a compact convex set. Given $R$, $q$ s.t. equation (\ref{eqn 18}) and (\ref{eqn 19}) which are linear. Hence $q$ is in a compact convex set. Given $R$, $F_1(q,R,U)$ is linear in $q$ and $U$. Then we can write,
\begin{align*}
\max_{R} \max_{U} \min_{q} & \sum_{l} q(l)\big(\nu(l)+U_{\mathcal{J}_1}\big)\\
& s.t.\ equation\ (\ref{eqn 16})\ to\ (\ref{eqn 20})
\end{align*}
A part of this LP, $\min_{q} \sum_{l} q(l)\big(\nu(l)+U_{\mathcal{J}_1}\big)$ s.t. $\sum_{l} q(l)=1$ and $q(l) \geq 0;\forall l$ can be written as the following dual LP.
\begin{align*}
\max_{U_0} & \ U_{0}\\
s.t.\ \ &\nu(l)+U_{\mathcal{J}_1} \geq U_{0},\ \ \forall l
\end{align*}
Thus LP (\ref{eqn new 19}) is computed. By using the same method we can prove the LP (\ref{eqn 14}). \qed
\end{proof}
The level of computational complexity of the LP for dual game is the same as the primal game LP.

\subsection{LPs to update sufficient statistic}
\label{LP algorithms to update vector payoffs in truncated finite horizon dual games}
In the window by window method, once we find the optimal strategies for the first window we need to update $\mu_{t+1}=\beta^*_{a_t,b_t}$ and $\nu_{t+1}=\alpha^*_{a_t,b_t}$ stage by stage, as given in Definition \ref{def: vector payoff} with some finite $n$. The optimal vector payoff $\beta^*$ and $\alpha^*$ can be computed by solving the LPs given in the following theorem.
\begin{theorem}
\label{LP2 of dual game}
Consider an $n$-stage type $1$ dual game $\tilde{\Gamma}_{n,\lambda}^1(\mu,q)$. Let $Y^*$ be player $2$'s optimal strategy at stage $1$ in $\tilde{\Gamma}_{n,\lambda}^1(\mu,q)$. Its game value satisfies
\begin{align}
 w^1_{n, \lambda}(\mu, q) &= \min_{Z_0^{a\in \mathcal{A},b\in \mathcal{B}}} \min_{Z_{\mathcal{I}}^{a\in \mathcal{A},b\in \mathcal{B}}} \min_{S^{a\in \mathcal{A},b\in \mathcal{B}}} \min_{\beta_{a\in \mathcal{A},b\in \mathcal{B}}} \min_{ \rho } \rho \label{LP2 object}\\
& s.t. \nonumber \\
& \rho \geq \mu (k) + \sum_{l,b} G_{k,l} (a,b) Y^* (b,l) q(l) + \lambda \sum_{l,b} Y^*(b,l) q(l)\Big(Z_0^{a,b} - \sum_{k'}  \nonumber\\
& P_{a,b}(k,k')\beta_{a,b}(k') \Big);\forall a,k \label{LP2_constraint_1}\\
&\beta_{a,b}(k) + Z_{\mathcal{I}_1}^{a,b} \leq Z_0^{a,b}; \ \ \forall k, \mathcal{I}_1=\lbrace k \rbrace, \forall a,b \label{LP2_constraint_2} \\
& \sum_{l_1,...l_s} \sum_{b_s} S_{\mathcal{J}_s}^{a,b} (b_s) \lambda^{s-1} G_{k_s,l_s}(a_s,b_s)+ \sum_{k_{s+1}} \sum_{b_s}P_{a_s,b_s}(k_s,k_{s+1}) Z_{\mathcal{I}_{s+1}}^{a,b}  \nonumber \\
& \leq Z_{\mathcal{I}_s}^{a,b}; \forall s=1,...n-1, \forall a_s, \forall \mathcal{I}_{s+1} \supset \mathcal{I}_s, \forall a,b \label{LP2_constraint_3}
\end{align}
\begin{align}
& \sum_{b_t} S_{\mathcal{J}_t}^{a,b} (b_t) = Q_{a_{t-1},b_{t-1}}(l_{t-1},l_t) S_{\mathcal{J}_{t-1}}^{a,b}(b_{t-1}); \forall t=2,...,n-1, \forall\mathcal{J}_t, \forall a,b \label{LP2_constraint_4}\\
&\sum_{b_1} S_{\mathcal{J}_1}^{a,b} (b_1)= q^+_{q,Y^*_{a,b}}(l_1); \ \ \forall \mathcal{J}_1, \forall a,b \label{LP2_constraint_5}\\
&Z_{\mathcal{I}_n}^{a,b}=0; \ \forall \mathcal{I}_n; \forall a,b \label{LP2_constraint_6} \\
& S^{a,b}_{\mathcal{J}_t}(b_t) \geq 0; \ \forall \mathcal{J}_t, t=1,...n-1, \forall b_t; \forall a,b \label{LP2_constraint_7}
\end{align}
The vector payoff in stage $2$ is, $\mu_2=\beta^*_{a_1,b_1}$. Similarly, consider an $n$-stage type $2$ dual game $\tilde{\Gamma}_{n,\lambda}^2(p, \nu)$. Let $X^*$ be player $1$'s optimal strategy in stage $1$ in $\tilde{\Gamma}_{n,\lambda}^2(p, \nu)$.  Its game value satisfies
\begin{align*}
&w^2_{n, \lambda}(p, \nu) = \max_{U_0^{a,b}} \max_{U_{\mathcal{J}}^{a,b}} \max_{R^{a,b}} \max_{\alpha_{a,b}} \max_{ \varphi } \varphi\\
& s.t.\\
& \varphi \leq \nu (l) + \sum_{k,a} G_{k,l} (a,b) X^*(a,k) p(k) + \lambda \sum_{k,a}X^*(a,k) p(k)\Big( U_0^{a,b}- \sum_{l'} Q_{a,b}(l,l') \\
&  \alpha_{a,b}(l') \Big);\forall b,l \\
&\alpha_{a,b}(l) + U_{\mathcal{J}_1}^{a,b} \geq U_0^{a,b}; \forall l, \mathcal{J}_1=\lbrace l \rbrace, \forall a,b\\
& \sum_{k_1,...k_s} \sum_{a_s} R_{\mathcal{I}_s}^{a,b} (a_s) \lambda^{s-1} G_{k_s,l_s}(a_s,b_s)+ \sum_{l_{s+1}} \sum_{a_s}Q_{a_s,b_s}(l_s,l_{s+1}) U_{\mathcal{J}_{s+1}}^{a,b} \geq U_{\mathcal{J}_s}^{a,b};\\
& \forall s=1,...n-1, \forall b_s, \forall \mathcal{J}_{s+1} \supset \mathcal{J}_s,\forall a,b\\
& \sum_{a_t} R_{\mathcal{I}_t}^{a,b} (a_t) = P_{a_{t-1},b_{t-1}}(k_{t-1},k_t) R_{\mathcal{I}_{t-1}}^{a,b}(a_{t-1});\forall t=2,...,n-1,\forall\mathcal{J}_t, \forall a,b\\
&\sum_{a_1} R_{\mathcal{I}_1}^{a,b} (a_1)= p^+_{p,X^*_{a,b}}(k_1); \  \forall \mathcal{I}_1, \forall a,b\\
&U_{\mathcal{J}_n}^{a,b}=0; \ \forall \mathcal{J}_n; \forall a,b \\
& R^{a,b}_{\mathcal{I}_t}(a_t) \geq 0; \ \forall \mathcal{I}_t, t=1,...n-1, \forall a_t; \forall a,b
\end{align*}
The vector payoff in stage $2$ is, $\nu_2=\alpha_{a_1,b_1}^*$.
\end{theorem}
The computational complexity of the LPs in Theorem \ref{LP2 of dual game} is the same as that of the primal game LP. We conclude this section with the detailed algorithm of our sufficient statistic based strategies as in Algorithm \ref{Algorithm_p1} and \ref{Algorithm_p2}.
\begin{algorithm}
\label{infinite horizon algorithm}
\caption{Sufficient statistic based strategy algorithm for long finite Bayesian game (Player 1)}
\label{Algorithm_p1}
Initialize $\lambda, p$ and $q$. Set window size $n$.\\
Set $t=1$. Let $p_1=p$ and $q_1=q$. \\
Run Theorem \ref{lp of primal game} with parameter $(p,q,n,\lambda)$ to get the optimal solution $\sigma^*$, and $\alpha^*$. The sub-optimal strategies of player $1$ from stage $1$ to $n$ is $X_{t:n}^{\star}(k)=\sigma_{1:n}^{\star}(k)$. The initial vector payoff is $\nu_1=-\alpha^*$.\\
Player $1$ takes action according to the sub-optimal strategy $X_{t}^{\star}(\mathcal{I}_t)$.\\
Record the action pair $(a_{t},b_{t})$.\\
Run Theorem \ref{LP2 of dual game} with parameters $(\nu_t, p_t, n, \lambda)$ in type $2$ dual game and get the optimal solution $\alpha^*_{a_t, b_t}$. Set $\nu_{t+1}=\alpha^\star_{a_{t},b_{t}}$, and update $p_{t+1}$ using equation (\ref{eq: p+}).\\
Update $t=t+1$.\\
If a new window starts i.e. $t\ mod\ n==1$
update the strategy for the new window as follows. \linebreak
(\RN{1}) If $N-t<n$ then $n=N-t+1$ \linebreak
(\RN{2}) Run Proposition \ref{prop: LP, dual game} with parameters $(\nu_t, p_t, n, \lambda)$ in type $2$ dual game, and get the optimal strategy $\sigma^\star$. Let the suboptimal strategy of player 1 is $X_{t:(t-1+n)}^\star=\sigma^{\star}_{1:n}$\\
If $t-1+n \neq N$ then go to \textit{Step 4}
\end{algorithm}
\begin{algorithm}
\label{infinite horizon algorithm}
\caption{Sufficient statistic based strategy algorithm for long finite Bayesian game (Player 2)}
\label{Algorithm_p2}
Initialize $\lambda, p$ and $q$. Set window size $n$.\\
Set $t=1$. Let $p_1=p$ and $q_1=q$. \\
Run Theorem \ref{lp of primal game} with parameter $(p,q,n,\lambda)$ to get the optimal solution $\tau^*$ and $\beta^*$. The sub-optimal strategy of player $2$ from stage $1$ to $n$ is $Y_{t:n}^{\star}(l)=\tau_{1:n}^{\star}(l)$. The initial vector payoffs of the type 1 dual game is $\mu_1=-\beta^*$.\\
Player $2$ takes action according to the sub-optimal strategy $Y_{t}^{\star}(\mathcal{J}_t)$.\\
Record the action pair $(a_{t},b_{t})$.\\
Run Theorem \ref{LP2 of dual game} with parameters $(\mu_t, q_t, n, \lambda)$ in type $1$ dual game and get the optimal solution $\beta^*_{a_t, b_t}$. Set $\mu_{t+1}=\beta^\star_{a_{t},b_{t}}$, and update $q_{t+1}$ using equation (\ref{eq: q+}).\\
Update $t=t+1$.\\
If a new window starts i.e. $t\ mod\ n==1$
update the strategy for the new window as follows. \linebreak
(\RN{1}) If $N-t<n$ then $n=N-t+1$ \linebreak
(\RN{2}) Run Proposition \ref{prop: LP, dual game} with parameters $(\mu_t, q_t, n, \lambda)$ in type $1$ dual game and get the optimal strategy $\tau^\star$. Let the suboptimal strategy $Y_{t:(t-1+n)}^\star=\tau^{\star}_{1:n}$\\
If $t-1+n \neq N$ then go to \textit{Step 4}
\end{algorithm}
\section{Performance Analysis}
\label{performance analysis}
In this section, we will analyze the performance of Algorithm \ref{Algorithm_p1}, \ref{Algorithm_p2} and observe how accurate the suboptimal strategies of the players are. For this we compare the worst case payoff of these suboptimal strategies with the game value. The worst case performance of player 1 is
\begin{align*}
J_{N}^{\sigma}(p,q)=& \min_{\tau} E \Big( \sum_{t=1}^{N} \lambda^{t-1} G_{k_t,l_t}(a_t,b_t) \Big)
\end{align*}
and the worst case performance of player 2 is
\begin{align*}
J_{N}^{\tau}(p,q)=& \max_{\sigma} E \Big( \sum_{t=1}^{N} \lambda^{t-1} G_{k_t,l_t}(a_t,b_t) \Big)
\end{align*}
Let $\sigma^*$ and $\tau^*$ be the optimal strategies for $N$ stages and $\sigma^{\star}$ and $\tau^{\star}$ be the window-by-window strategies for $N$ stages. The worst case performance of $\sigma^{\star}$ is represented by $J_{N}^{\sigma^{\star}}(p,q)$ and the worst case performance of $\tau^{\star}$ is represented by $J_{N}^{\tau^{\star}}(p,q)$.

\begin{theorem}
\label{performance bound}
If player $2$ takes actions according to window-by-window strategy $\tau^{\star}$ then its worst case performance $J_{N}^{\tau^{\star}}(p,q)$ satisfies
\begin{align}
J_{N}^{\tau^{\star}}(p,q)-v_{N}(p,q) \leq \lambda^n \frac{1-\lambda^{N-n}}{1-\lambda} \bar{G} \label{performance measure of player 2}
\end{align}
If player $1$ takes actions according to window-by-window strategy $\sigma^{\star}$ then its worst case performance $J_{N}^{\sigma^{\star}}(p,q)$ satisfies
\begin{align}
v_{N}(p,q)-J_{N}^{\sigma^{\star}}(p,q) \leq \lambda^n \frac{1-\lambda^{N-n}}{1-\lambda} \bar{G} \label{performance measure of player 1}
\end{align}
where, $\bar{G}= \max_{k,l,a,b} G_{k,l}(a,b)$
\end{theorem}
\begin{proof}
The worst case performance of player $2$ by using the strategy from Algorithm \ref{Algorithm_p2} is
\begin{align*}
J_{N}^{\tau^{\star}}(p,q)=& \max_{\sigma} E \Big( \sum_{t=1}^{N} \lambda^{t-1} G_{k_t,l_t}(a_t,b_t) \Big)\\
\leq & \max_{\sigma} E \Big( \sum_{t=1}^{n} \lambda^{t-1} G_{k_t,l_t}(a_t,b_t) + \sum_{t=n+1}^{N} \lambda^{t-1} \bar{G} \Big)\\
=& \max_{\sigma} E \Big( \sum_{t=1}^{n} \lambda^{t-1} G_{k_t,l_t}(a_t,b_t) + \lambda^n(1+...+ \lambda^{N-n-1}) \bar{G} \Big)\\
=& \max_{\sigma} E \Big( \sum_{t=1}^{n} \lambda^{t-1} G_{k_t,l_t}(a_t,b_t) \Big) + \lambda^n(1+...+ \lambda^{N-n-1}) \bar{G}\\
=& v_n(p,q)+ \lambda^n \frac{1-\lambda^{N-n}}{1- \lambda} \bar{G}
\end{align*}
The last equation holds because the window-by-window strategy uses the optimal strategy in the first window. As $G$ has all non-negative elements, we can say
\begin{align*}
v_{n}(p,q) \leq v_{N}(p,q)
\end{align*}
Therefore,
\begin{align}
J_{N}^{\tau^{\star}}(p,q)- v_{N}(p,q) \leq J_{N}^{\tau^{\star}}(p,q)- v_{n}(p,q) \leq \lambda^n \frac{1-\lambda^{N-n}}{1- \lambda} \bar{G} \label{for performance of P2}
\end{align}
As the window-by-window strategy uses the optimal strategy in the first window and $G$ has all non-negative elements, we can say
\begin{align}
v_n(p,q) \leq J_N^{\sigma^{\star}}(p,q) \label{relation between vn and J sigma prime}
\end{align}
According to the definition of game value, we have
\begin{align}
J^{\sigma^{\star}}_{N}(p,q) \leq v_{N}(p,q) \label{relation between J sigma prime and vN}
\end{align}
From inequality (\ref{relation between vn and J sigma prime}) and (\ref{relation between J sigma prime and vN}) we can write 
\begin{align}
v_n(p,q) \leq J^{\sigma^{\star}}_{N}(p,q) \leq v_{N}(p,q) \leq J^{\tau^{\star}}_{N}(p,q) \label{relation between J sigam prime and J tau prime}
\end{align}
By using inequalities (\ref{for performance of P2}) and (\ref{relation between J sigam prime and J tau prime}) we get
\begin{align}
v_{N}(p,q)-J_{N}^{\sigma^{\star}}(p,q) \leq J^{\tau^{\star}}_{N}(p,q)-v_n(p,q) \leq \lambda^n \frac{1-\lambda^{N-n}}{1- \lambda} \bar{G} \label{for performance of P1}
\end{align}
which completes the proof. \qed
\end{proof}
We can see from theorem \ref{performance bound} that the difference of the game value and the worst case performance of Algorithm \ref{Algorithm_p1} and \ref{Algorithm_p2} decreases with the increasing window size. Meanwhile, the size of the LPs increases with respect to the window size. That's why we need to select the window size carefully so that we can get good performance without exhausting the computational resources. 

\section{Case study}
\label{case study}
Jamming problems in underwater acoustic sensor network is an practical example of strategic game explained in game theory. First, \cite{7417412} used jamming in an underwater sensor network as a Bayesian two-player zero-sum one-shot game and evaluated how the nodes' position (state) effects the equilibrium and \cite{8567999} extended it to a repeated Bayesian game with uncertainties on both sides. In this paper we have used the same network model and formulated it as a two-player zero-sum stochastic game where both players are partially informed.

The network has four sensors $(s_1, s_2, s_3, s_4)$ and one jammer. The goal of the sensors is to transmit data to a sink node by using a shared spectrum at $[10, 40]$ KHz and the goal of the jammer is to block this data transfer. This spectrum is divided into two channels, $B_1=[10,25]$ kHz and $B_2=[25,40]$ kHz. Each channel can be used by one sensor so that at a time only 2 sensors can transmit data. The distance between the sensor and sink node can be 1 km or 5 km. Let, the distances from sink node to $s_1, s_2, s_3$ and $s_4$ are 1 km, 5 km, 1 km and 5 km, respectively. The jammer distance from the sink node can be 0.5 km or 2 km. Let us consider sensors as player 1 and the jammer as player 2. The possible states for player 1 are [1,1], [1,5] and [5,5] and states for player 2 are [0.5] and [2]. For sensors, we are denoting [1,5] and [5,1] as the same state. The initial state probabilities of player 1 and 2 are $p=[0.5\ 0.3\ 0.2]$ and $q=[0.5\ 0.5]$, respectively.
The sensors coordinate with each other to use the channels so that they can maximize the data transmission. Channel 1 is more effective for the sensors which are far away and channel 2 is better for the sensors close by. Though the sensors and the jammer do not know each other’s position, the jammer can observe whether a channel is used by a far away sensor or a close by sensor. For each time period, the jammer can only generate noises in one channel and sensors are able to detect it. The jammer tries to minimize the throughput of the channels. If the sensor state is [1,1], the active sensors are $s_1$ and $s_3$. The feasible actions are $s_1$ using channel 1 while $s_3$ using channel 2 (action 1), and $s_1$ using channel 2 while $s_3$ using channel 1 (action 2). If the sensor state is [5,5], the active sensors are $s_2$ and $s_4$. The feasible actions are $s_2$ using channel 1 while $s_4$ using channel 2 (action 1), and $s_2$ using channel 2 while $s_4$ uses channel 1 (action 2). If the sensor state is [1,5], the feasible actions are faraway sensor using channel 1 while nearby sensor using channel 2 (action 1), and faraway sensor using channel 2 while nearby sensor using channel 1 (action 2). Similarly, the jammer's actions are blocking channel 1 which is action 1 or channel 2 which is action 2. The payoff matrix is in Table \ref{Table: payoff matrix}. As we are modeling this network as a stochastic game there will be transition matrices for both player's state. Table \ref{Table: P transition matrix} and \ref{Table: Q transition matrix} show the transition matrices of sensors' and jammer's state, respectively.

\begin{table}
\center
\caption{The total channel capacity matrix (Payoff matrix), given sensors' type $k$ and jammer's type $l$ For example, $G_{k_t=2, l_t=2}(a_t=1, b_t=1)=24.89$ }
\label{Table: payoff matrix}
\begin{tabular}{|c|cc|cc|}
  \hline
  \backslashbox{k}{l}  & \multicolumn{2}{|c}{1 (0.5 km)} & \multicolumn{2}{|c|}{2 (2 km)} \\ \hline
  \multirow{2}{*}{1 ([1 1] km)} & 108.89 & 113.78 & 122.30 & 154.40 \\
   & 108.89 & 113.78 & 122.30 & 154.40 \\ \hline
  \multirow{2}{*}{2 ([1 5] km)} & 11.48 & 107.38 & 24.89 & 107.42\\
   & 99.04 & 20.15 & 100.26 & 60.77 \\ \hline
  \multirow{2}{*}{3 ([5 5] km)} & 1.64 & 13.75 & 2.85 &13.79 \\
   & 1.64 & 13.75 & 2.85 &13.79 \\
  \hline
\end{tabular}\\
\end{table}
\begin{table}
\center
\caption{Transition matrix for sensor's state (P) For example, $P_{a=1,b=1}(k_{t-1}=2,k_t=2)=0.4$ }
\label{Table: P transition matrix}
\begin{tabular}{|c|ccc|ccc|}
  \hline
\backslashbox{a}{b}  & \multicolumn{3}{|c}{1} & \multicolumn{3}{|c|}{2} \\ \hline
\multirow{3}{*}{1} & 0.8   & 0.1   & 0.1   & 0.4   & 0.5   & 0.1 \\
                   & 0.1   & 0.4   & 0.5   & 0.2   & 0.3   & 0.5 \\
                   & 0.2   & 0.7   & 0.1   & 0.4   & 0.4   & 0.2 \\ \hline
\multirow{3}{*}{2} & 0.2   & 0.2   & 0.6   & 0.3   & 0.3   & 0.4 \\
                   & 0.5   & 0.2   & 0.3   & 0.1   & 0.8   & 0.1 \\
                   & 0.2   & 0.2   & 0.6   & 0.1   & 0.1   & 0.8 \\  \hline
\end{tabular}
\end{table}
\begin{table}
\center
\caption{Transition matrix for jammer's state (Q) For example, $Q_{a=1,b=1}(l_{t-1}=2,l_t=2)=0.5$ }
\label{Table: Q transition matrix}
\begin{tabular}{|c|cc|cc|}
\hline
\backslashbox{a}{b}  & \multicolumn{2}{|c}{1} & \multicolumn{2}{|c|}{2} \\ \hline
\multirow{2}{*}{1}  & 0.8   & 0.2   & 0.2   & 0.8 \\
                    & 0.5   & 0.5   & 0.1   & 0.9 \\ \hline
\multirow{2}{*}{2}  & 0.6   & 0.4   & 0.7   & 0.3 \\
                    & 0.5   & 0.5   & 0.1   & 0.9 \\  \hline
\end{tabular}
\end{table}

\subsubsection{Performance bound check}
Because the size of LPs grows exponentially in time horizon, we can only compute the game value for the cases with $N<=4$. For this reason, we set $N=4$ with window size, $n=2$ and $\lambda=0.3$. 

To check the performance bound for player 2, we use extreme case where player 1 is playing according to its optimal strategy which it gets from the primal game LP and player 2 is playing according to window-by-window method (Algorithm \ref{Algorithm_p2}). From Table \ref{Table: payoff matrix} we get $\bar{G}=154.4$. The game value of this game is $v_N(p,q)=112.9049$. According to theorem \ref{performance bound} player 2's performance should satisfy the following condition.
\begin{align}
J_{4}^{\tau^{\star}}(p,q)-112.9049 &\leq 0.3^2 \times \frac{1-0.3^{4-2}}{1-0.3} \times 154.4=18.07 \label{eqn for p2 worst bound check}
\end{align}
We run this game for $500$ times and get the average payoff, $J_{N}^{\tau^{\star}}(p,q)=116.27$ which satisfies equation (\ref{eqn for p2 worst bound check}).

Similarly, for player 1's performance bound check, player 1 is using window-by-window method (Algorithm \ref{Algorithm_p1}) and player 2 is playing optimally. According to theorem \ref{performance bound}, for this game, player 2's performance should satisfy the following condition.
\begin{align}
112.9049-J_{4}^{\sigma^{\star}}(p,q) \leq 0.3^2 \times \frac{1-0.3^{4-2}}{1-0.3} \times 154.4= 18.07
\end{align}
After running the game for $500$ times we get the average payoff $J_{N}^{\sigma^{\star}}(p,q)=105.49$ which satisfies the above condition. From this we can see, even in extreme cases player 1 and 2 satisfy their performance bound.
\subsubsection{Performance comparison of different window size}
To analyze the effect of window size in the players' performance we set a $36$ stage game and run it for window size $n=2$ and $n=3$ using $\lambda$ from $0.1$ to $0.9$. Player 1's goal is to maximize the payoff and player 2's goal is to minimize it. 

To evaluate the effect in player 1's performance we run this game with player 1 playing according to Algorithm \ref{Algorithm_p1} and player 2 playing according to a fixed strategy such that when $l_t=1$ player 2 takes action $1$ and $2$ with probability $[0.9\ 0.1]$ and when $l_t=2$ player 2 takes action $1$ and $2$ with probability $[0.75\ 0.25]$. We run the game for $500$ times and see that (Table \ref{tab:payoff for different window}) for each value of $\lambda$ the bigger window has higher payoff which means player 1's performance is getting better with increasing window size. 

For player 2's performance we run the same game with player 1 playing fixed strategy such that for $k_t=1,2$ and $3$ its probability of taking action 1 and 2 is $[0.9\ 0.1]$, $[0.75\ 0.25]$ and $[0.95\ 0.05]$, respectively and player 2 playing according to Algorithm \ref{Algorithm_p2}. The average payoff we get by running the game for $500$ times is in Table \ref{tab:payoff for different window}. We can see for $\lambda=0.1\ to\ 0.4$ the payoff is lower for the large window which means player 2's performance is getting better with increasing window size. But for $\lambda>0.4$ this effect is not visible as with increasing $\lambda$ the payoff increases exponentially. It is clear that both players' performance get better with increasing window size. But we have to keep in mind that large window also increase computational complexity. 
\begin{table}
\centering
\caption{Payoff for different window size}
\label{tab:payoff for different window}
\resizebox{\textwidth}{!}{%
\begin{tabular}{|c|c|c|c|c|c|c|c|c|c|}
\hline
\textbf{$\lambda$} & \textbf{0.1} & \textbf{0.2} & \textbf{0.3} & \textbf{0.4} & \textbf{0.5} & \textbf{0.6} & \textbf{0.7} & \textbf{0.8} & \textbf{0.9} \\ \hline
\textbf{P1 Optimal Payoff (n=2)}  & 87.38        & 100.93       & 120.44       & 139.29       & 160.78       & 196.74       & 250.33       & 364.06       & 677.79       \\ \hline
\textbf{P1 Optimal Payoff (n=3)}  & 91.47        & 104.03       & 124.08       & 144.85       & 178.09       & 215.63       & 274.19       & 388.49       & 750.85       \\ \hline
\textbf{P2 Optimal Payoff (n=2)}  & 89.95        & 98.24        & 114.04       & 130.65       & 148.99       & 185.30       & 244.05       & 355.51            & 689.66            \\ \hline
\textbf{P2 Optimal Payoff (n=3)}  & 85.65        & 96.40        & 113.84       & 127.98       & 153.62       & 190.55       & 244.29       & 365.12            & 708.04          \\ \hline
\end{tabular}%
}
\end{table}
\section{Conclusion}
This paper analyzes two-player zero-sum stochastic Bayesian games. The traditional methods used for computing the optimal strategies in Bayesian games are computationally heavy and cannot perform well for long horizon games. This paper introduces a computationally efficient Algorithm which can solve long Bayesian games and it's sufficient statistic is fully accessible to the players. It solves the game window-by-window and update players' sufficient statistic using LP, based on which players can compute their optimal strategies. Thus it solves the accessibility issue which is a major problem of Bayesian game. The performance of this Algorithm is within the bound and gets better with increasing window size. 
\section{Appendix}
\begin{lemma}
\label{Lemma 4.2}
The weighted payoffs $ U_{\mathcal{J}_t}(\sigma,\tau) $ and $Z_{\mathcal{I}_t}(\sigma,\tau)$ satisfy the following recursive formulas.
\begin{align}
U_{\mathcal{J}_t}(\sigma,\tau)=& \sum_{b_t} \Big(\sum_{a_t} \sum_{k_1,...k_t} R_{\mathcal{I}_t}(a_t)\lambda^{t-1}G_{k_t,l_t}(a_t,b_t)+ \sum_{a_t} \sum_{l_{t+1}} Q_{a_t,b_t}(l_t,l_{t+1}) \nonumber \\
& U_{\mathcal{J}_{t+1}}(\sigma,\tau)\Big) \tau_{t}^{b_t}(\mathcal{J}_t); \ \ \ with\ U_{ \mathcal{J}_{N+1}}(\sigma,\tau)=0 \label{eqn 9} \\
Z_{\mathcal{I}_t}(\sigma,\tau)=& \sum_{a_t} \Big(\sum_{b_t} \sum_{l_1,...l_t} S_{\mathcal{J}_t}(b_t)\lambda^{t-1} G_{k_t,l_t}(a_t,b_t)+ \sum_{b_t} \sum_{k_{t+1}} P_{a_t,b_t}(k_t,k_{t+1}) \nonumber \\
& Z_{\mathcal{I}_{t+1}}(\sigma,\tau)\Big) \sigma_{t}^{a_t}(\mathcal{I}_t);\ \ \ with\ Z_{\mathcal{J}_{N+1}}(\sigma,\tau)=0
\end{align}
\end{lemma}
\begin{proof}
According to the definition, we have
\begin{align*}
U_{\mathcal{J}_{t+1}}(\sigma,\tau)=& \sum_{k_1,...k_{t+1}} R_{\mathcal{I}_{t}}(a_{t})P_{a_{t},b_{t}}(k_{t},k_{t+1}) E\Big(\sum_{s=t+1}^{N}\lambda^{s-1}G_{k_s,l_s}(a_s,b_s)|k_1,...k_{t+1},\mathcal{J}_{t+1} \Big)\\
U_{\mathcal{J}_{t}}(\sigma,\tau)=& \sum_{k_1,...k_t} R_{\mathcal{I}_{t-1}}(a_{t-1})P_{a_{t-1},b_{t-1}}(k_{t-1},k_{t}) E\Big(\sum_{s=t}^{N} \lambda^{s-1} G_{k_s,l_s}(a_s,b_s)|k_1,...k_t,\mathcal{J}_t \Big)\\
=& \sum_{k_1,...k_t} R_{\mathcal{I}_{t-1}}(a_{t-1})P_{a_{t-1},b_{t-1}}(k_{t-1},k_{t}) E\Big(\lambda^{t-1} G_{k_t,l_t}(a_t,b_t)|k_1,...k_t,\mathcal{J}_t \Big)+\\
& \sum_{k_1,...k_t} R_{\mathcal{I}_{t-1}}(a_{t-1}) P_{a_{t-1},b_{t-1}}(k_{t-1},k_{t}) E\Big(\sum_{s=t+1}^{N}\lambda^{s-1}G_{k_s,l_s}(a_s,b_s)|k_1,...k_t,\mathcal{J}_t \Big)\\
Term\ 1:& \sum_{k_1,...k_t} R_{\mathcal{I}_{t-1}}(a_{t-1})P_{a_{t-1},b_{t-1}}(k_{t-1},k_{t}) E\Big(\lambda^{t-1} G_{k_t,l_t}(a_t,b_t)|k_1,...k_t,\mathcal{J}_t \Big)\\
=& \sum_{k_1,...k_t} R_{\mathcal{I}_{t-1}}(a_{t-1})P_{a_{t-1},b_{t-1}}(k_{t-1},k_{t})\sum_{a_t,b_t} \lambda^{t-1}G_{k_t,l_t}(a_t,b_t) Pr(a_t,b_t|k_1,...k_t,\mathcal{J}_t)\\
=& \sum_{k_1,...k_t} R_{\mathcal{I}_{t-1}}(a_{t-1})P_{a_{t-1},b_{t-1}}(k_{t-1},k_{t})\sum_{a_t,b_t}\lambda^{t-1} G_{k_t,l_t}(a_t,b_t)\sigma_{t}^{a_t}(\mathcal{I}_t) \tau_{t}^{b_t}(\mathcal{J}_t)\\
=& \sum_{a_t,b_t} \sum_{k_1,...k_t} R_{\mathcal{I}_{t}}(a_{t}) \lambda^{t-1} G_{k_t,l_t}(a_t,b_t)\tau_{t}^{b_t}(\mathcal{J}_t)
\end{align*}
\begin{align*}
Term\ 2 :& \sum_{k_1,...k_t} R_{\mathcal{I}_{t-1}}(a_{t-1})P_{a_{t-1},b_{t-1}}(k_{t-1},k_{t}) E\Big(\sum_{s=t+1}^{N} \lambda^{s-1} G_{k_s,l_s}(a_s,b_s)|k_1,...k_t,\mathcal{J}_t \Big)\\
=& \sum_{k_1,...k_t} R_{\mathcal{I}_{t-1}}(a_{t-1}) P_{a_{t-1},b_{t-1}}(k_{t-1},k_{t})\sum_{k_{t+1},l_{t+1}} \sum_{a_t,b_t} P_{a_t,b_t}(k_t,k_{t+1})Q_{a_t,b_t}(l_t,l_{t+1})  \\
& \sigma_{t}^{a_t}(\mathcal{I}_t) \tau_{t}^{b_t}(\mathcal{J}_t) E\Big( \sum_{s=t+1}^{N} \lambda^{s-1} G_{k_s,l_s}(a_s,b_s)|k_1,...k_{t+1},\mathcal{J}_{t+1}\Big)\\
=& \sum_{l_{t+1},a_t,b_t} \Big(\sum_{k1,...k_{t+1}} R_{\mathcal{I}_t}(a_t)P_{a_t,b_t}(k_t,k_{t+1}) E(\sum_{s=t+1}^{N}  \lambda^{s-1} G_{k_s,l_s}(a_s,b_s)|k_1,...k_{t+1},
\end{align*}
\begin{align*}
& \mathcal{J}_{t+1})\Big)Q_{a_t,b_t}(l_t,l_{t+1}) \tau_{t}^{b_t}(\mathcal{J}_t)\\
=& \sum_{a_t,b_t} \sum_{l_{t+1}} Q_{a_t,b_t}(l_t,l_{t+1}). U_{\mathcal{J}_{t+1}}.\tau_{t}^{b_t}(\mathcal{J}_t)
\end{align*}
So, equation (\ref{eqn 9}) is proved. Following the similar method, we can prove the recursive formula of $ Z_{\mathcal{I}_t}(\sigma,\tau) $. \qed
\end{proof}
\subsection{\textbf{Proof of Lemma \ref{lemma: LP, U star, Z star}}}
\begin{proof}
Let us define two optimization problems $P1$ and $P2$.
\begin{align}
P1: \bar{U}_{\mathcal{J}_t}(\sigma)=& \max_{U} U_{\mathcal{J}_t} \nonumber \\
s.t.
& \sum_{k_1...k_s} \sum_{a_s} R_{\mathcal{I}_s}(a_s) \lambda^{s-1}G_{k_s,l_s}(a_s,b_s) + \sum_{l_{s+1}} \nonumber \\
& \sum_{a_s} Q_{a_s,b_s}(l_s,l_{s+1}) U_{\mathcal{J}_{s+1}} \geq U_{\mathcal{J}_s}; \quad \forall s=t,...N \label{eqn 11}
\end{align}
\noindent We know,
\begin{align*}
U^{*}_{\mathcal{J}_t}(\sigma) =& \min_{\tau_{t}} \sum_{b_t} \Big(\sum_{a_t} \sum_{k_1,...k_t} R_{\mathcal{I}_t}(a_t)\lambda^{t-1} G_{k_t,l_t}(a_t,b_t)+ \\
& \sum_{a_t} \sum_{l_{t+1}} Q_{a_t,b_t}(l_t,l_{t+1})U_{\mathcal{J}_{t+1}}^{*}(\sigma)\Big) \tau_{t}^{b_t}(\mathcal{J}_t);\\
s.t. \
& 1^T \tau_{N}^{b_t}(\mathcal{J_N}) =1\\
& \tau_{N}^{b_t}(\mathcal{J_N}) \geq 0 ;\ \ \forall b_N   \\
\end{align*}
The dual of this LP is,
\begin{align*}
P2: U_{\mathcal{J}_t}^{*}(\sigma)=& \max_{ U_{\mathcal{J}_t}} U_{\mathcal{J}_t} \\
s.t.
& \sum_{k_1...k_t} \sum_{a_t} R_{\mathcal{I}_t}(a_t)\lambda^{t-1}G_{k_t,l_t}(a_t,b_t)+ \sum_{l_{t+1}} \sum_{a_t}\\
& Q_{a_t,b_t}(l_t,l_{t+1}) U^{*}_{\mathcal{J}_{t+1}} \geq U_{\mathcal{J}_t};\ \  \forall b_t
\end{align*}
\noindent For $ t=N $,
\begin{align*}
\bar{U}_{\mathcal{J}_N}(\sigma) =& \max_{U} U_{\mathcal{J}_N}\\
& s.t.\\
\sum_{k_1...k_N} \sum_{a_N}& R_{\mathcal{I}_N}(a_N) \lambda^{N-1}G_{k_N,l_N}(a_N,b_N) \geq U_{\mathcal{J}_N}(\sigma);\ \ \forall b_N
\end{align*}
\begin{align*}
U_{\mathcal{J}_N}^{*}(\sigma)=& \max_{U} U_{\mathcal{J}_N} \\
& s.t.\\
\sum_{k_1...k_N} \sum_{a_N} & R_{\mathcal{I}_N}(a_N)\lambda^{N-1}G_{k_N,l_N}(a_N,b_N) \geq U_{\mathcal{J}_N}(\sigma); \ \ \forall b_N
\end{align*}
\noindent From this we can see, for $t=N$, $\bar{U}_{\mathcal{J}_N}(\sigma)= U_{\mathcal{J}_N}^{*}(\sigma)$. Let us assume, $U^{*}_{\mathcal{J}_{t+1}}= \bar{U}_{\mathcal{J}_{t+1}} $.
\begin{align*}
P3: U^{*}_{\mathcal{J}_{t+1}}(\sigma)=& \max_{U} U_{\mathcal{J}_{t+1}}\\
s.t.
& \sum_{k_1...k_s} \sum_{a_s} R_{\mathcal{I}_s}(a_s) \lambda^{s-1} G_{k_s,l_s}(a_s,b_s)+ \sum_{l_{s+1}} \sum_{a_s} Q_{a_s,b_s}(l_s,l_{s+1}) U_{\mathcal{J}_{s+1}}\\
& \geq U_{\mathcal{J}_s};\ \ \forall s=(t+1)...N, \forall b_s,\mathcal{J}_s\supset J_{t+1}
\end{align*}
Let,$ \bar{U}_{\mathcal{J}_t}, \bar{U}_{\mathcal{J}_{t+1}},...\bar{U}_{\mathcal{J}_N}$ be the optimal solution of $P1$. $U^{*}_{\mathcal{J}_t}$ be the optimal solution of $P2$.$ \bar{U}^{*}_{\mathcal{J}_{t+1}},...\bar{U}^{*}_{\mathcal{J}_N}$ be the optimal solution of $P3$. Here, $U^{*}$ is feasible in $P1$. So, $\bar{U}_{\mathcal{J}_t}\geq U^{*}_{\mathcal{J}_t} $. $\bar{U}_{\mathcal{J}_{t+1}}...\bar{U}_{\mathcal{J}_N}$ is feasible in $P3$. So, $U^{*}_{\mathcal{J}_{t+1}}\geq \bar{U}_{\mathcal{J}_{t+1}}$. For $s=t$, we can write equation (\ref{eqn 11}) as,
\begin{align*}
& \sum_{k_1...k_t} \sum_{a_t} R_{\mathcal{I}_t}(a_t) \lambda^{t-1}G_{k_t,l_t}(a_t,b_t) + \sum_{l_{t+1}}
\sum_{a_t} Q_{a_t,b_t}(l_t,l_{t+1})\bar{U}_{\mathcal{J}_{t+1}}(\sigma) \geq U_{\mathcal{J}_t}(\sigma);\\
& \therefore \sum_{k_1...k_t} \sum_{a_t} R_{\mathcal{I}}(a_t) \lambda^{t-1}G_{k_t,l_t}(a_t,b_t) + \sum_{l_{t+1}} \sum_{a_t} Q_{a_t,b_t}(l_t,l_{t+1}) U^{*}_{\mathcal{J}+1}(\sigma) \geq U_{\mathcal{J}_t}(\sigma)
\end{align*}
So, $\bar{U}_{\mathcal{J}_t}$ is feasible in $P2$. Then $U^{*}_{\mathcal{J}_t}\geq \bar{U}_{\mathcal{J}_t}$.Therefore, $ \bar{U}_{\mathcal{J}_{t}}= U^{\ast}_{\mathcal{J}_{t}}$, which completes the proof.Similarly we can proof the LP for $Z^*_{\mathcal{I}_t}$. \qed
\end{proof}
\subsection{\textbf{Proof of Theorem \ref{LP2 of dual game}}}
\begin{proof}
The recursive formula of dual game equation (\ref{dual_recursive_type_1}) can be written as,
\begin{align}
w^1_{n, \lambda}(\mu, q) =& \min_Y \min_{\beta_{a\in \mathcal{A},b \in \mathcal{B}}} \max_{\Pi} \sum_{a,k} \Pi(a,k) \Big[ \mu (k) + \sum_{l,b} G_{k,l} (a,b) \nonumber \\
& Y(b,l) q(l) + \lambda \sum_{l,b} Y(b,l) q(l) \Big( w^1_{n-1 , \lambda} (\beta_{a,b}, q^+_{q,Y}) \nonumber \\
& - \sum_{k'} P_{a,b}(k,k') \beta_{a,b}(k') \Big) \Big] \label{recursive formula of dual game}
\end{align}
From the LP in Proposition \ref{lp of primal game}, we can find the optimal strategy of player $2$, $Y^*$ for $n=1$ stage. Then we can write the above equation as,
\begin{align*}
w^1_{n, \lambda}(\mu, q)
=& \min_{\beta_{a \in \mathcal{A},b \in \mathcal{B}}} \max_{\Pi} \sum_{a,k} \Pi(a,k) \Big[ \mu (k) + \sum_{l,b} G_{k,l} (a,b) Y^* \\
&(b,l) q(l) + \lambda \sum_{l,b} Y^*(b,l) q(l) \Big( w^1_{n-1 , \lambda} (\beta_{a,b}, q^+_{q,Y^*}) - \\
& \sum_{k'} P_{a,b}(k,k') \beta_{a,b}(k') \Big) \Big]
\end{align*}
The dual LP of the above equation is as follows.
\begin{align}
w^1_{n, \lambda}(\mu, q)=& \min_{\beta_{a \in \mathcal{A},b \in \mathcal{B}}} \min_{ \rho } \rho \label{eqn dual game value LP}\\
&s.t. \nonumber \\
& \rho \geq \mu (k) + \sum_{l,b} G_{k,l} (a,b) Y^*(b,l)q(l) + \lambda \sum_{l,b} Y^*(b,l)q(l) \Big( w^1_{n-1 , \lambda} (\beta_{a,b}, \nonumber \\
& q^+_{q,Y_{a,b}^*})- \sum_{k'} P_{a,b}(k,k')\beta_{a,b}(k') \Big); \forall a,k \nonumber \\
\Rightarrow& \rho \geq \mu (k) + \sum_{l,b} G_{k,l} (a,b) Y^*(b,l) q(l) + \lambda \sum_{l,b} Y^*(b,l)q(l) \Big( Z_0^* (\beta_{a,b}, q^+_{q,Y_{a,b}^*}) \nonumber \\
& - \sum_{k'} P_{a,b}(k,k')\beta_{a,b}(k') \Big); \forall a,k \nonumber
\end{align}
where for any $a\in \mathcal{A}$ and $b\in \mathcal{B}$,
\begin{align*}
Z_0^* (\beta_{a,b}, q^+_{q,Y^*_{a,b}})=& \min_{Z_0^{a,b}} \min_{Z_{\mathcal{I}}^{a,b}} \min_{S^{a,b}} Z_0^{a,b}\\
&s.t.\\
&\beta_{a,b}(k) + Z_{\mathcal{I}_1}^{a,b} \leq Z_0^{a,b};\ \ \forall k, \mathcal{I}_1=\lbrace k \rbrace \\
& \sum_{l_1,...l_s} \sum_{b_s} S_{\mathcal{J}_s}^{a,b} (b_s) \lambda^{s-1} G_{k_s,l_s}(a_s,b_s)+ \sum_{k_{s+1}} \sum_{b_s} P_{a_s,b_s}(k_s,k_{s+1}) Z_{\mathcal{I}_{s+1}}^{a,b}  \\
& \leq Z_{\mathcal{I}_s}^{a,b}; \forall s=1:n-1, \forall a_s, \forall \mathcal{I}_{s+1} \supset \mathcal{I}_s\\
& \sum_{b_t} S_{\mathcal{J}_t}^{a,b} (b_t) = Q_{a_{t-1},b_{t-1}}(l_{t-1},l_t) S_{\mathcal{J}_{t-1}}^{a,b} (b_{t-1}); \forall t=2,...,n-1; \forall\mathcal{J}_t\\
&\sum_{b_1} S_{\mathcal{J}_1}^{a,b} (b_1)= q^+_{q,Y^*_{a,b}}(l_1);\ \forall \mathcal{J}_1\\
&Z_{\mathcal{I}_n}^{a,b}=0;\ \forall \mathcal{I}_n\\
& S^{a,b}_{\mathcal{J}_t}(b_t) \geq 0;\ \forall \mathcal{J}_t,b_t,t=1,...n-1
\end{align*}
We can split the LP (\ref{eqn dual game value LP}) in two parts $P_1$ and $P_3$.
\begin{align}
P_1: &w^1_{n, \lambda}(\mu, q)= \min_{\beta_{a \in \mathcal{A},b \in \mathcal{B}}} \min_{ \rho } \rho\\
& s.t. \nonumber \\
& \rho \geq \mu (k) + \sum_{l,b} G_{k,l} (a,b) Y^*(b,l) q(l) + \lambda \sum_{l,b} Y^*(b,l)q(l) \Big( Z_0^* (\beta_{a,b}, q^+_{q,Y_{a,b}^*})- \nonumber \\
&  \sum_{k'} P_{a,b}(k,k')\beta_{a,b}(k') \Big); \forall a,k \label{A}\\
P_3: &Z_0^* (\beta_{a,b}, q^+_{q,Y^*_{a,b}})= \min_{Z_0^{a,b}} \min_{Z_{\mathcal{I}}^{a,b}} \min_{S^{a,b}} Z_0^{a,b}; \ \ \forall a,b\\
& s.t.\\
&\beta_{a,b}(k) + Z_{\mathcal{I}_1}^{a,b} \leq Z_0^{a,b};\ \ \forall k, \mathcal{I}_1=\lbrace k \rbrace \label{B} \\
& \sum_{l_1,...l_s} \sum_{b_s} S_{\mathcal{J}_s}^{a,b} (b_s) \lambda^{s-1} G_{k_s,l_s}(a_s,b_s) \nonumber + \sum_{k_{s+1}} \sum_{b_s} P_{a_s,b_s}(k_s,k_{s+1}) Z_{\mathcal{I}_{s+1}}^{a,b} \leq Z_{\mathcal{I}_s}^{a,b};\nonumber
\end{align}
\begin{align}
& \forall s=1,...n-1, \forall a_s, \forall \mathcal{I}_{s+1} \supset \mathcal{I}_s \label{C}\\ 
& \sum_{b_t} S_{\mathcal{J}_t}^{a,b} (b_t) = Q_{a_{t-1},b_{t-1}}(l_{t-1},l_t) S_{\mathcal{J}_{t-1}}^{a,b}(b_{t-1});\ \  \forall t=2,...,n-1; \forall\mathcal{J}_t \label{D} \\
&\sum_{b_1} S_{\mathcal{J}_1}^{a,b} (b_1)= q^+_{q,Y^*_{a,b}}(l_1);\ \ \forall \mathcal{J}_1 \label{E} \\
&Z_{\mathcal{I}_n}^{a,b}=0;\ \ \forall \mathcal{I}_n \label{F}\\
& S^{a,b}_{\mathcal{J}_t}(b_t) \geq 0;\ \ \forall \mathcal{J}_t,b_t,t=1,...n-1 \label{G}
\end{align}
Let,
\begin{align}
P_2: &w^1_{n, \lambda}(\mu, q) \min_{Z_0^{a \in \mathcal{A},b \in \mathcal{B}}} \min_{Z_{\mathcal{I}}^{a \in \mathcal{A},b \in \mathcal{B}}} \min_{S^{a \in \mathcal{A},b \in \mathcal{B}}} \min_{\beta_{a \in \mathcal{A},b \in \mathcal{B}}} \min_{ \rho } \rho\\
& s.t.\\
& \rho \geq \mu (k) + \sum_{l,b} G_{k,l} (a,b) Y^*(b,l) q(l) + \lambda \sum_{l,b}Y^*(b,l) q(l) \Big( Z_0^{a,b}- \sum_{k'} P_{a,b}(k,k') \nonumber \\
& \beta_{a,b}(k') \Big); \forall a,k\label{i}\\
&\beta_{a,b}(k) + Z_{\mathcal{I}_1}^{a,b} \leq Z_0^{a,b}; \ \ \forall k, \mathcal{I}_1=\lbrace k \rbrace, \forall a,b \label{j} \\
& \sum_{l_1,...l_s} \sum_{b_s} S_{\mathcal{J}_s}^{a,b} (b_s) \lambda^{s-1} G_{k_s,l_s}(a_s,b_s)+ \sum_{k_{s+1}} \sum_{b_s}P_{a_s,b_s}(k_s,k_{s+1}) Z_{\mathcal{I}_{s+1}}^{a,b} \leq Z_{\mathcal{I}_s}^{a,b}; \nonumber\\
&   \forall s=1,...n-1, \forall a_s, \forall \mathcal{I}_{s+1} \supset \mathcal{I}_s, \forall a,b \label{k} \\
& \sum_{b_t} S_{\mathcal{J}_t}^{a,b} (b_t) = Q_{a_{t-1},b_{t-1}}(l_{t-1},l_t) S_{\mathcal{J}_{t-1}}^{a,b}(b_{t-1}); \forall t=2,...,n-1; \forall\mathcal{J}_t, \forall a,b \label{l}\\
&\sum_{b_1} S_{\mathcal{J}_1}^{a,b} (b_1)= q^+_{q,Y^*_{a,b}}(l_1);\ \ \forall \mathcal{J}_1, \forall a,b \label{m}\\
&Z_{\mathcal{I}_n}^{a,b}=0; \ \ \forall \mathcal{I}_n; \forall a,b \label{n} \\
& S^{a,b}_{\mathcal{J}_t}(b_t) \geq 0;\ \ \forall \mathcal{J}_t,b_t,t=1,...n-1, \forall b_t; \forall a,b \label{o}
\end{align}
Let the optimal solution for $P_2$ be $\bar{\beta}_{a,b}, \bar{S}_{\mathcal{J}}^{a,b}, \bar{Z}_{\mathcal{I}}^{a,b}, \bar{Z_0}^{a,b}, \bar{\rho}$. It is easy to verify that the optimal solutions of $P_2$ are feasible solutions in $P_3$ because constraints (\ref{i}-\ref{o}) are same as (\ref{B}-\ref{G}).
\begin{align}
\therefore \bar{Z}^{a,b}_0 \geq Z_0^* (\beta_{a,b}, q^+_{q,Y^*_{a,b}}) \label{relation_P2_P3}
\end{align}
$P_1$ has $2$ variables $\beta, \rho$ and $1$ constraint equation (\ref{A}). $P_2$'s optimal solution must satisfy its constraint (\ref{i}). Hence,
\begin{align*}
\bar{\rho} &\geq \mu (k) + \sum_{l,b} G_{k,l} (a,b) Y^*(b,l) q(l) + \lambda \sum_{l,b} Y^*(b,l) q(l)\\
& \Big( \bar{Z}_0^{a,b} - \sum_{k'} P_{a,b}(k,k')\bar{\beta}_{a,b}(k') \Big);\ \forall a,k
\end{align*}
Equation (\ref{relation_P2_P3}) implies that
\begin{align*}
 \bar{\rho} &\geq \mu (k) + \sum_{l,b} G_{k,l} (a,b) Y^*(b,l) q(l) + \lambda \sum_{l,b} Y^*(b,l) q(l) \\
&\Big( Z_0^* (\bar{\beta}_{a,b}, q^+_{q,Y_{a,b}^*})- \sum_{k'} P_{a,b}(k,k')\bar{\beta}_{a,b}(k') \Big);\ \forall a,k
\end{align*}
The equation above is $P_1$'s constraint (\ref{A}). So, $P_2$'s optimal solution is feasible in $P_1$. If $\rho^*$ is the optimal solution for $P_1$ then,
\begin{align}
\bar{\rho} \geq \rho^* \label{rho bar geq rho*}
\end{align}
Let the optimal solution of $P_1$ be $\beta_{a,b}^*, \rho^*$ and its nested LP $P_3$ has optimal solution ${S_{\mathcal{J}}^{a,b}}^*, {Z_{\mathcal{I}}^{a,b}}^*, {Z^{a,b}_0}^*$. It is easy to verify that $\beta_{a,b}^*, \rho^*, {S_{\mathcal{J}}^{a,b}}^*, {Z_{\mathcal{I}}^{a,b}}^*, {Z^{a,b}_0}^*$ are feasible in $P_2$. So we have $\rho^* \geq \bar{\rho}$. Together with equation (\ref{rho bar geq rho*}), we know that
\begin{align*}
\rho^*= \bar{\rho}
\end{align*}
It completes the proof of equation (\ref{LP2 object}-\ref{LP2_constraint_7}). Similarly, we can show the LP for type 2 dual game $w^2_{n, \lambda}(p, \nu)$.
\end{proof}

\bibliographystyle{abbrv}
\bibliography{myref}
\end{document}


\title{Supplementary Documents of LP Formulations of Two-Player Zero-Sum Stochastic Bayesian Finite Horizon games
}


\author{Nabiha Nasir Orpa         \and
        Lichun Li 
}


\institute{Nabiha Nasir Orpa and Lichun Li \at
              Department of Industrial and Manufacturing Engineering, FAMU-FSU College of Engineering, 2525 Pottsdamer St, Tallahassee, FL 32310, USA\\
              \email{no18k@my.fsu.edu, lichunli@eng.famu.fsu.edu}          
}

\date{Received: date / Accepted: date}

\maketitle

\section*{Proofs of propositions, lemmas and theorems}
\label{sec: primal game}
\subsubsection{Game value}
\label{lemma: primal, concavity}
In primal game $ \Gamma_{N,\lambda}(p,q) $, the game value $v_{N,\lambda}(p,q)$ is concave in $ p $ and convex in $ q $.
\begin{proof}
Let, initial probability, $ p=\alpha_w p_w $ where $ W\in \lbrace1,2\rbrace $. Here, $p,p_w \in \Delta{K}$, $ \alpha_{w}\in\Delta{W} $ and $ \sigma^{w} $ is optimal for player $1$ in game $\Gamma_{n,\lambda}(p_w,q)$. In initial state $ k $, player $1$ uses $ \sigma^{w}(k,:) $ with probability $ \alpha_w\frac{p_{w}(k)}{p(k)} $. Then the payoff is,
\begin{align*}
\gamma_{n,\lambda} (\delta_{k},\delta_{l},\sigma(k),\tau(l)) =& E_{\sigma(k),\tau(l)}\Big(\sum\limits_{t=1}^n \lambda^{t-1}  G_{k_{t},l_{t}}(a_{t},b_{t})|k,l\Big)\\
=&E_w\Big[E_{\sigma^{w}(k),\tau(l)}\Big(\sum\limits_{t=1}^n \lambda^{t-1} G_{k_{t},l_{t}}(a_{t},b_{t})|k,l,w\Big)\Big]\\
=&\sum\limits_{w} Pr(w|k,l) \Big[E_{\sigma^{w}(k),\tau_{l}}\Big(\sum\limits_{t=1}^n \lambda^{t-1} G_{k_{t},l_{t}}(a_{t},b_{t})|k,l,w\Big)\Big]\\
=&\sum\limits_{w} Pr(w|k)\gamma_{n,\lambda} (\delta_{k},\delta_{l},\sigma^{w}(k),\tau_{l})\\
=&\sum\limits_{w} \alpha_{w}\frac{p_{w}(k)}{p(k)}\gamma_{n,\lambda} (\delta_{k},\delta_{l},\sigma^{w}(k),\tau_{l})\\
\gamma_{n,\lambda}(p,q,\sigma,\tau) =&\sum\limits_{k,l} p(k) q(l) \gamma_{n,\lambda} (\delta_{k},\delta_{l},\sigma(k),\tau(l)) \\
=&\sum\limits_{k,l} p(k) q(l)\sum\limits_{w} \alpha_{w}\frac{p_{w}(k)}{p(k)}\gamma_{n,\lambda} (\delta_{k},\delta_{l},\sigma^{w}(k),\tau_{l})
\end{align*}
\begin{align*}
=&\sum\limits_{k,l} p(k) q(l)\sum\limits_{w} \alpha_{w}\frac{p_{w}(k)}{p(k)}\gamma_{n,\lambda} (\delta_{k},\delta_{l},\sigma^{w}(k),\tau_{l})\\
=&\sum\limits_{k,l} q(l)\sum\limits_{w} \alpha_{w}p_{w}(k)\gamma_{n,\lambda} (\delta_{k},\delta_{l},\sigma^{w}(k),\tau_{l})\\
=&\sum\limits_{w} \alpha_{w}\sum\limits_{k,l}p_{w}(k)q(l)\gamma_{n,\lambda} (\delta_{k},\delta_{l},\sigma^{w}(k),\tau_{l}) \\
=&\sum\limits_{w} \alpha_{w}\gamma_{n,\lambda} (p_{w},q,\sigma^{w},\tau)
\end{align*}
This is true for any $ \tau $. Let, $\tau^{*}$ is the optimal strategy in $\Gamma_{n,\lambda}(p,q)$. So,
\begin{align*}
\sum\limits_{w} \alpha_{w}\gamma_{n,\lambda} (p_{w},q,\sigma^{w},\tau^{\ast})&\geqslant \sum\limits_{w} \alpha_{w} v_{n,\lambda}(p_{w},q)\\
\Rightarrow \gamma_{n,\lambda} (p,q,\sigma,\tau^*)&\geqslant \sum\limits_{w} \alpha_{w} v_{n,\lambda}(p_{w},q)\\
\Rightarrow v_{n,\lambda}(p,q) &\geqslant \sum\limits_{w} \alpha_{w} v_{n,\lambda}(p_{w},q); &[As\ v_{n,\lambda}(p,q) \geq \gamma_{n,\lambda} (p,q,\sigma,\tau^*)]
\end{align*}
From this we can say that $v_{n,\lambda}(p,q)$ is concave on $ p $. By following the similar method me can prove that $v_{n,\lambda}(p,q)$ is convex on $ q $.
\qed
\end{proof}
For the simplicity of the mathematical expression of Lemma \ref{lemma: primal, recursive}, we define
\begin{align*}
\bar{G}(p,q,X,Y)= & \sum\limits_{k,l} p(k)q(l)\sum\limits_{a,b}X(a,k)G_{k,l}(a,b)Y(b,l) \\ 
\mathfrak{T}_{p,q,X,Y}(v)=&  \bar{G}(p,q,X,Y)+\lambda \sum_{a\in \mathcal{A}}\sum_{b\in \mathcal{B}} \bar{x}_{p,X}(a) \bar{y}_{q,Y}(b) v(p^{+}_{p,X},q^{+}_{q,Y})
\end{align*} 
\begin{lemma}
\label{lemma: primal, recursive}
Consider a primal game $ \Gamma_{N,\lambda}(p,q) $. Let $p_t,q_t$ be the believes at stage $t=1,\ldots,N$ with $n=N+1-t$ stages to go. The game value of the primal game satisfies the following recursive formula
\begin{align}
v_{n,\lambda}(p_t,q_t)=& \max\limits_{X\in\Delta \mathcal{(A)}^\mathcal{K}} \min\limits_{Y\in \Delta (\mathcal{B})^\mathcal{L}} \mathfrak{T}_{p_t,q_t,X,Y}(v_{n-1,\lambda}) \label{eq: primal, recursive, player 1}\\
=&\min\limits_{Y\in \Delta (\mathcal{B})^\mathcal{L}} \max\limits_{X\in\Delta (\mathcal{A})^\mathcal{K}} \mathfrak{T}_{p_t,q_t,X,Y}(v_{n-1,\lambda}) \label{eqn 2}
\end{align}
The optimal solution $X^*$ to equation (\ref{eq: primal, recursive, player 1}) and the optimal solution $Y^*$ to equation (\ref{eqn 2}) are player 1 and 2's security strategies at stage $t$, respectively.

Moreover, player 1 and 2's security strategies only depend on $p_t$ and $q_t$, together with $t$ if $N$ is finite.
\end{lemma}
\begin{proof}
The payoff of the primal game can be written as,
\begin{align*}
&\gamma_{n,\lambda}(p,q,\sigma,\tau)= \bar{G}(p,q,X,Y)+ E_{p,q,\sigma,\tau}(\sum\limits_{t'=2}^{n}\lambda^{t'-1} G_{k_{t'}l_{t'}}(a_{t'},b_{t'})) \\
&E_{p,q,\sigma,\tau}\Big(\sum_{t'=2}^{n}\lambda^{t'-1} G_{k_{t'}l_{t'}}(a_{t'},b_{t'})\Big)\\
=& E\Big[E\Big(\sum_{t'=2}^{n}\lambda^{t'-1} G_{k_{t'}l_{t'}}(a_{t'},b_{t'})| k_2=k,l_2=l,a_1=a,b_1=b\Big)\Big]\\
=&\sum_{k_2,l_2,a_1,b_1} E\Big(\sum_{t'=2}^{n}\lambda^{t'-1} G_{k_{t'}l_{t'}}(a_{t'},b_{t'})|k_2,l_2,a_1,b_1\Big) Pr(k_2,l_2,a_1,b_1)\\
=&\sum_{k_2,l_2,a_1,b_1}E\Big(\sum_{t'=2}^{n}\lambda^{t'-1} G_{k_{t'}l_{t'}}(a_{t'},b_{t'})| k_2,l_2,a_1,b_1\Big) \sum_{k_1,l_1} Pr(k_2,l_2,a_1,b_1,k_1,l_1)\\
=&\sum_{k_2,l_2,i_1,j_1}E(\sum_{t'=2}^{n}\lambda^{t'-1} G_{k_{t'}l_{t'}}(a_{t'},b_{t'})| k_2,l_2,a_1,b_1) \sum_{k_1,l_1} Pr(k_2,l_2|a_1,b_1,k_1,l_1)\\
& Pr(a_1,b_1,k_1,l_1)\\
=&\sum_{k_2,l_2,a_1,b_1}E(\sum_{t'=2}^{n}\lambda^{t'-1} G_{k_{t'}l_{t'}}(a_{t'},b_{t'})| k_2,l_2,a_1,b_1)\sum_{k_1} Pr(k_2|a_1,b_1,k_1)Pr(a_1,k_1)\\
& \sum_{l_1} Pr(l_2|a_1,b_1,l_1) Pr(b_1,l_1) \\
=&\sum_{k_2,l_2,a_1,b_1} E(\sum_{t'=2}^{n}\lambda^{t'-1} G_{k_{t'}l_{t'}}(a_{t'},b_{t'})| k_2,l_2,a_1,b_1) \sum_{k_1} Pr(k_2|a_1,b_1,k_1)\\
&\frac{Pr(a_1,k_1)}{\sum_{k_1} Pr(a_1,k_1)}\sum_{k_1} Pr(a_1,k_1) \sum_{l_1} Pr(l_2|a_1,b_1,l_1) \frac{Pr(b_1,l_1)}{\sum_{l_1} Pr(b_1,l_1)} \sum_{l_1}Pr(b_1,l_1)\\
=&\sum_{a_1,b_1}\sum_{k_2,l_2} \sum_{k_1} Pr(k_2|a_1,b_1,k_1)\frac{Pr(a_1,k_1)}{\sum_{k_1} Pr(a_1,k_1)} \sum_{l_1} Pr(l_2|a_1,b_1,l_1)\frac{Pr(b_1,l_1)}{\sum_{l_1} Pr(b_1,l_1)} \\
&E\Big(\sum_{t=1}^{n-1} \lambda^{t}G_{k_t,l_t}(a_t,b_t)|k_2,l_2,a_1,b_1\Big) \sum_{k_1}Pr(a_1,k_1)\sum_{l_1}Pr(b_1,l_1); [Here\ t=t'-1]\\
=&\lambda\sum_{a_1,b_1} \gamma_{n-1, \lambda}(p^{+},q^{+},\sigma,\tau) \sum_{k_1} Pr(a_1,k_1)  \sum_{l_1}Pr(b_1,l_1)\\
=&\lambda\sum_{a_1,b_1} \Big[\sum_{k_1} Pr(a_1|k_1)Pr(k_1)\Big] \Big[\sum_{l_1}Pr(b_1|l_1)Pr(l_1)\Big] \gamma_{n-1, \lambda}(p^{+},q^{+},\sigma,\tau) \\
=&\lambda\sum_{a_1} \overline{x}_{p,X}(a) \sum_{b_1}\overline{y}_{q,Y}(b) \gamma_{n-1, \lambda}(p^{+},q^{+},\sigma,\tau) 
\end{align*}
So we can write,

\begin{align*}
\gamma_{n,\lambda}(p,q,\sigma,\tau)=& \bar{G}(p,q,X,Y)+\lambda \sum\limits_{a_1} \sum\limits_{b_1} \bar{x}_{p,X}(a) \bar{y}_{q,Y}(b) \gamma_{n-1, \lambda}(p^{+},q^{+},\sigma,\tau)
\end{align*}
Let us define an auxiliary zero-sum game $\hat{\Gamma}_{n,\lambda}(p,q)$ where the payoff function for player 1 is defined by, 
\begin{align*}
f_{n,\lambda}^{p,q}=\bar{G}(p,q,X,Y)+ \lambda \sum\limits_{a\in \mathcal{A}} \sum\limits_{b\in \mathcal{B}}\bar{x}_{p,X}(a) \bar{y}_{q,Y}(b) v_{n-1,\lambda}(p^{+},q^{+})
\end{align*}
By following the similar proof as \cite{2226507620060801} we can say, $ f_{n,\lambda}^{p,q} $ is concave in $ X $ and convex in $ Y $. Given $X$, $\bar{G}(p,q,X,Y)$ is linear w.r.t. $Y$ and given $Y$, it is linear w.r.t. $X$. Let, $X=\Theta X' + (1- \Theta) X''$. As $v_{n-1,\lambda}$ is concave in $p$,
\begin{align*}
v_{n-1,\lambda}(\sum_{a,b} \bar{x}_{p,X}(a)\bar{y}_{q,Y}(b)p^{+},q^{+})\geq & \Theta v_{n-1,\lambda}(\sum_{a,b} \bar{x}'_{p,X}(a)\bar{y}_{q,Y}(b)p_{X'}^{+},q^{+})+(1-\Theta)\\
& v_{n-1,\lambda}(\sum_{a,b} \bar{x}''_{p,X}(a)\bar{y}_{q,Y}(b)p_{X''}^{+},q^{+})
\end{align*} 
As $v_{n-1, \lambda}(\alpha p, q)= \alpha v_{n-1, \lambda}(p, q) $, we can write the above equation as,
\begin{align*}
&\sum_{a,b} \bar{x}_{p,X}(a)\bar{y}_{q,Y}(b) v_{n-1,\lambda}(p^{+},q^{+}) \geq \Theta \sum_{a,b}\bar{x}'_{p,X'}(a)\bar{y}_{q,Y}(b) v_{n-1,\lambda}(p_{X'}^{+},q^{+})+(1-\Theta)\\
&\sum_{a,b} \bar{x}''_{p,X''}(a)\bar{y}_{q,Y}(b) v_{n-1,\lambda}(p_{X''}^{+},q^{+})\\
&\Rightarrow \bar{G}(p,q,X,Y)+\sum_{a,b} \bar{x}_{p,X}(a)\bar{y}_{q,Y}(b) v_{n-1,\lambda}(p^{+},q^{+})\geq \Theta \Big(\bar{G}(p,q,X',Y)+\sum_{a,b} \bar{x}'_{p,X'}(a)\\
&\bar{y}_{q,Y}(b) v_{n-1,\lambda}(p_{X'}^{+},q^{+})\Big)+ (1-\Theta)\Big((\bar{G}(p,q,X'',Y)+ \sum_{a,b} \bar{x}''_{p,X''}(a)\bar{y}_{q,Y}(b) v_{n-1,\lambda}(p_{X''}^{+},q^{+})\Big)\\
&\Rightarrow f_{n,\lambda}^{p,q}(X,Y) \geq \Theta f_{n,\lambda}^{p,q}(X',Y) + (1- \Theta)f_{n,\lambda}^{p,q}(X'',Y) 
\end{align*}
From this we can say $f_{n,\lambda}^{p,q}(X,Y)$ is concave in $X$. By using the similar steps we can prove $f_{n,\lambda}^{p,q}(X,Y)$ is convex in $Y$. So, according to \cite{sion1958general}, $\hat{\Gamma}_{N,\lambda}(p,q)$ has a value. Let, the value is $f_{n,\lambda}(p,q)$. In $\hat{\Gamma}_{N,\lambda}(p,q)$ at stage 1, $ X^{*} $ and $ Y^{*} $ is optimal for player 1 and 2 respectively. We will now prove that player 1 can guarantee $f_{n,\lambda}(p,q)$ in $\Gamma_{n,\lambda}(p,q)$. Let, $\sigma \in \Sigma$ be as follows:
At stage 1, play $X^{*}$. From stage $2$ to $n$, play the optimal strategy $\sigma_a$ of stage 2 to $n$ in $\Gamma_{n,\lambda}(p,q)$. Let, $\tau \in \mathcal{T}$ and denote by $Y$ in $\Delta (\mathcal{B})$ the mixed action played by $\tau$ at stage 1 and for each $(a,b) \in \mathcal{A}\times \mathcal{B}$, by $\tau_{a,b}$ the strategy played by player 2 at stage 2 to $n$ if $(a,b)$ is played at stage 1. Then
\begin{align}
\gamma_{n,\lambda}(p,q,\sigma,\tau)=& \bar{G}(p,q,X^{*},Y)+\lambda \sum\limits_{a_1,b_1} \bar{x}^{*}_{p,X^{*}}(a) \bar{y}_{q,Y}(b) \gamma_{n-1}(p_{X^{*}}^{+},q^{+},\sigma,\tau) \label{41} \\ 
f_{n,\lambda}^{p,q}(X^*,Y)=&\bar{G}(p,q,X^*,Y)+ \lambda \sum\limits_{a,b}\bar{x}^{*}_{p,X*}(a) \bar{y}_{q,Y}(b) v_{n-1,\lambda}(p^{+}_{X^*},q^{+}) \label{42}
\end{align}
For 1 stage game,
\begin{align}
\gamma_{n-1,\lambda}(p_{X^{*}}^{+},q^{+},\sigma_{a},\tau) \geq v_{n-1,\lambda}(p_{X^{*}}^{+},q^{+})\label{43}
\end{align}
From equation (\ref{41}),(\ref{42}), (\ref{43}) we can write, $\gamma_{n,\lambda}(p,q,\sigma,\tau) \geq f_{n,\lambda}^{p,q}(X^{*},Y)$ and we know $ f_{n,\lambda}^{p,q}(X^{*},Y) \geq f_{n,\lambda}(p,q)$.
\begin{align*}
&\therefore \gamma_{n,\lambda}(p,q,\sigma,\tau) \geq f_{n,\lambda}^{p,q}(X^{*},Y) \geq f_{n,\lambda}(p,q);\ [This\ is\ true\ for\ any\ \tau]\\
&\Rightarrow \gamma_{n,\lambda}(p,q,\sigma,\tau^{*}) \geq f_{n,\lambda}^{p,q}(X^{*},Y) \geq f_{n,\lambda}(p,q)\\ 
&\therefore v_{n}(p,q) \geq f_{n}(p,q);\ [As\ v_{n}(p,q) \geq \gamma_{n}(p,q,\sigma,\tau^{*})]
\end{align*}
By using the similar method we can prove $ v_{n}(p,q)\leqslant f_{n}(p,q)$. That means, $ v_{n}(p,q)=f_{n}(p,q)$. \qed
\end{proof}
\begin{proposition}
\label{proposition 3.1}
Let $v_{n,\lambda}(p,q)$ be the game value of the primal game $\Gamma_{n,\lambda}(p,q)$, and $w_{n,\lambda}^1(\mu,q)$ and $w^{2}_{n,\lambda}(p,\nu)$ be the game values of type 1 and 2 dual games, respectively. We have
\begin{align}
v_{n,\lambda}(p,q)=&\min_{\mu} w^{1}_{n,\lambda}(\mu,q)-p^{T}\mu \label{eqn 3}\\
v_{n,\lambda}(p,q)=&\max_{\nu} w^{2}_{n,\lambda}(p,\nu)-q^{T}\nu \label{eqn 4}\\
w^{1}_{n,\lambda}(\mu,q)=& \max_{p} p^{T}\mu +v_{n,\lambda}(p,q) \label{eqn 5}\\
w^{2}_{n,\lambda}(p,\nu)= &\min_{q} q^{T}\nu +v_{n,\lambda}(p,q) \label{eqn 6}
\end{align}
Moreover, player 2's optimal strategy in type 1 dual game $\Tilde{\Gamma}^{1}_{n,\lambda}(\mu^{*},q)$ is also its optimal strategy in primal game $\Gamma_{n,\lambda}(p,q)$, where $\mu^*$ is the optimal solution to equation (\ref{eqn 3}) and player 1's optimal strategy in type 2 dual game $\Tilde{\Gamma}^{2}_{n,\lambda}(p,\nu^{*})$ is also its optimal strategy in primal game $\Gamma_{N,\lambda}(p,q)$, where $\nu^*$ is the optimal solution to equation (\ref{eqn 4}).
\end{proposition}
\begin{proof}
For type 1: Let, $ \tilde{\sigma}^{*} $ and $ \tilde{\tau}^{*} $ be the optimal strategy of player 1 and 2 respectively in type 1 dual game.
\begin{align}
&\tilde{\gamma}_{N,\lambda}^{1}(\mu,q,p,\sigma,\tilde{\tau}^{*})\leq \tilde{\gamma}_{N,\lambda}^{1}(\mu,q,p,\tilde{\sigma}^{*},\tilde{\tau}^{*}) \nonumber \\
\Rightarrow & \tilde{\gamma}_{N,\lambda}^{1}(\mu,q,p,\sigma,\tilde{\tau}^{*})\leq w^{1}_{N,\lambda}(\mu,q) \nonumber \\
\Rightarrow & E_{p,q,\sigma,\tilde{\tau}^{*}}\Big(\mu(k_1)+\sum\limits_{t=1}^{N} \lambda^{t-1} G_{k_t,l_t}(a_t,b_t)\Big) \leq w^{1}_{N,\lambda}(\mu,q) \nonumber \\
\Rightarrow & \sum\limits_{k_1}\mu(k_1)Pr(k_1)+ \gamma_{N,\lambda}(p,q,\sigma,\tilde{\tau}^{*})\leq w^{1}_{N,\lambda}(\mu,q) \nonumber \\
\Rightarrow & p^T\mu + \gamma_{N,\lambda}(p,q,\sigma,\tilde{\tau}^{*})\leq w^{1}_{N,\lambda}(\mu,q) \nonumber \\
\Rightarrow & \gamma_{N,\lambda}(p,q,\sigma^{*},\tilde{\tau}^{*})\leq w^{1}_{N,\lambda}(\mu,q)-p^T\mu \label{eqn 8}\\
\Rightarrow & \gamma_{N,\lambda}(p,q,\sigma^{*},\tau^{*})\leq w^{1}_{N,\lambda}(\mu,q)-p^T\mu ; \nonumber \\
\Rightarrow & v_{N,\lambda}(p,q)\leq w^{1}_{N,\lambda}(\mu,q)-p^T\mu;  [for\ any\ p,\ q\ and\ \mu]  \nonumber\\
& v_{N,\lambda}(p,q)\leq w^{1}_{N,\lambda}(\mu,q)-p^T\mu  
\end{align}
Let, $\tau^*$ be the optimal strategy of player 2 in primal game $\Gamma_{N,\lambda}(p,q)$. Let us define, $\beta(\tau^{*})= \max_{\sigma_k} E_{\sigma,\tau^{*},q}(\sum\limits_{t=1}^N \lambda^{t-1} G_{k_t,l_t}(a_t,b_t)|k_1=k)$
\begin{align}
\label{eqn:p_t_beta}
\begin{split}
p^T\beta(\tau^*)&= \sum\limits_{k_1} Pr(k_1)\max_{\sigma_k} E_{\sigma,\tau^{*},q}(\sum\limits_{t=1}^N \lambda^{t-1} G_{k_t,l_t}(a_t,b_t)|k_1=k)\\
&= \max_{\sigma}\sum\limits_{k_1} Pr(k_1) E_{\sigma,\tau^*,q}(\sum\limits_{t=1}^N \lambda^{t-1} G_{k_t,l_t}(a_t,b_t)|k_1=k)\\
&= \max_{\sigma}E_{p}\Big[E_{\sigma,\tau^*,q}\Big(\sum\limits_{t=1}^N \lambda^{t-1} G_{k_t,l_t}(a_t,b_t)|k_1=k\Big)\Big]\\
&= \max_{\sigma}E_{p,\sigma,\tau^*,q}\Big(\sum\limits_{t=1}^N \lambda^{t-1} G_{k_t,l_t}(a_t,b_t)\Big)\\
&= v_{N,\lambda}(p,q)
\end{split}
\end{align}
In dual game $ \tilde{\Gamma}^{1}(-\beta(\tau^*),q) $,
\begin{align*}
\tilde{\gamma}^{1}_{N,\lambda}(\mu,q,p,\sigma,\tau^*)&= E_{p}\Big(-\beta(\tau^{*})(k_1)\Big)+ E_{\sigma,\tau^{*},p,q}\Big(\sum\limits_{t=1}^N \lambda^{t-1} G_{k_t,l_t}(a_t,b_t)\Big)\\
&=\sum\limits_{k_1} -\beta(\tau^*)(k_1) Pr(k_1)+ \gamma_{N,\lambda}(p,q,\sigma,\tau^{*})\\
&=p^T [-\beta_k(\tau^*)] + \gamma_{N,\lambda}(p,q,\sigma,\tau^{*})\\
&=-v_{N,\lambda}(p,q)+\gamma_{N,\lambda}(p,q,\sigma,\tau^{*})\\
\gamma_{N,\lambda}(p,q,\sigma^{*},\tau^{*}) &\geq \gamma_{N,\lambda}(p,q,\sigma,\tau^{*})\\
\Rightarrow v_{N,\lambda}(p,q)&\geq \gamma_{N,\lambda}(p,q,\sigma,\tau^{*})
\end{align*}
So, we can say,
\begin{align*}
&\tilde{\gamma}^{1}_{N,\lambda}(\mu,q,p,\sigma,\tau^{*}) \leq 0; [For\ any\ \sigma,p,q]\\
&\tilde{\gamma}^{1}_{N,\lambda}(\mu,q,p,\tilde{\sigma^{*}},\tau^{*}) \leq 0\\
&\Rightarrow w^{1}_{N,\lambda}(-\beta(\tau^{*}),q) \leq 0\\
&\Rightarrow w^{1}_{N,\lambda}(-\beta(\tau^{*}),q) \leq v_{N,\lambda}(p,q)+ p^T(-\beta(\tau^{*}))\\
&\Rightarrow v_{N,\lambda}(p,q) \geq w^{1}_{N,\lambda}(-\beta(\tau^{*}),q)-[p^T(-\beta(\tau^{*})]\\
&\Rightarrow v_{N,\lambda}(p,q) \geq w^{1}_{N,\lambda}(\mu,q)-[p^T\mu]
\end{align*}
So we can say for some specific $ \mu=-\beta(\tau^*) $ in type 1 dual game,
\begin{align}
\label{eqn:10}
v_{N,\lambda}(p,q)&= w^{1}_{N,\lambda}(\mu,q)- p^T\mu
\end{align}
\noindent In $\tilde{\Gamma}^{1}_{n,\lambda}(\mu,q)$ the optimal strategy of player 2 is $\tilde{\tau}^{*}$. So from equation (\ref{eqn 8}) we get,
\begin{align*}
\gamma_{n,\lambda}(p,q,\sigma^*,\tilde{\tau}^{*}) &\leq w_{n,\lambda}^{1}(\mu, q) -p^{T}\mu ; \quad \forall \mu \\
\gamma_{n,\lambda}(p,q,\sigma^*,\tilde{\tau}^{*}) &\geq v_{n,\lambda}(p,q) 
\end{align*}
If $\mu=-\beta(\tau^*)$ we can write
\begin{align*}
& \gamma_{n,\lambda}(p,q,\sigma^*,\tilde{\tau}^{*}) \leq v_{n,\lambda}(p,q)\\
\therefore & \gamma_{n,\lambda}(p,q,\sigma^*,\tilde{\tau}^{*})= v_{n,\lambda}(p,q)
\end{align*}
$\therefore \tilde{\tau}^*$ is also the optimal in $\Gamma_{n,\lambda}(p,q)$. Similarly, we can prove the equation for type 2 dual game. According to equation (\ref{eqn:p_t_beta}) $p^T \beta(\tau^*)=v_{N,\lambda}(p,q)$. Then for $ \mu=-\beta(\tau^*) $ we can write equation (\ref{eqn:10}) as,
\begin{align*}
v_{n,\lambda}(p,q)&= w^{1}_{n,\lambda}(\mu,q)- p^T\mu\\
\Rightarrow v_{n,\lambda}(p,q)&= w^{1}_{n,\lambda}(\mu,q)+ v_{n,\lambda}(p,q)\\
\Rightarrow w^{1}_{n,\lambda}(\mu^*,q)&=0
\end{align*}
Similarly, we can prove $w^{2}_{n,\lambda}(p,\nu^*)=0$.

For type 1 dual game,
\begin{align*}
w^{1}_{n,\lambda}(\mu,q)=& \max_{p} \max_{\sigma} \min_{\tau} \tilde{\gamma}^{1}(\mu,q,p,\sigma,\tau)\\
=& \max_{p} \max_{\sigma} \min_{\tau} E_{p,q,\sigma,\tau}\Big(\mu(k_1)+ \sum\limits_{t=1}^{N}\lambda^{t-1} G_{k_t,l_t}(a_t,b_t)\Big)\\
=& \max_{p} \max_{\sigma} \min_{\tau} E_{p}\Big(\mu(k_1)\Big)+E_{p,q,\sigma,\tau}\Big(\sum\limits_{t=1}^{N}\lambda^{t-1}G_{k_t,l_t}(a_t, b_t)\Big)\\
=& \max_{p} \sum\limits_{k_1} \mu(k_1)Pr(k_1)+ \max_{\sigma} \min_{\tau} E_{p,q,\sigma,\tau}\Big(\sum\limits_{t=1}^{N} \lambda^{t-1} G_{k_t,l_t}(a_t,b_t)\Big)\\
=& \max_{p} p^T \mu + \max_{\sigma} \min_{\tau}  \gamma_{n,\lambda}(p,q,\sigma,\tau)\\
=& \max_{p} p^T \mu + v_{n,\lambda}(p,q)
\end{align*}
By similar technique we can prove equation (\ref{eqn 6}) for type 2 dual game. \qed
\end{proof}
In the primal game $\Gamma_{n,\lambda}(p,q)$, define
\begin{align*}
\alpha_{l}(\sigma) =& \min_{\tau} E_{p,q,\sigma,\tau}\Big( \sum_{t=1}^{N} \lambda^{t-1}G_{k_t,l_t}(a_t,b_t)| l_1=l \Big)\\
\beta_{k}(\tau) =& \max_{\sigma} E_{p,q,\sigma,\tau}\Big( \sum_{t=1}^{N} \lambda^{t-1}G_{k_t,l_t}(a_t,b_t)| k_1=k \Big)
\end{align*}
Let $\alpha(\sigma)= (\alpha_l(\sigma))_{l\in L}$ and $\beta(\tau)=(\beta_k(\tau))_{k\in K}$. From the proof of Proposition \ref{proposition 3.1}, we also show the following corollary.
\begin{corollary}
\label{corollary: 3.2}
An optimal solution to equation (\ref{eqn 3}) and equation (\ref{eqn 4}) is 
\begin{align}
\mu^*=&-\beta(\tau^*) \label{eq: mu star} \\
\nu^*=&-\alpha(\sigma^*) \label{eq: nu star} 
\end{align} 
where $\sigma^*$ and $\tau^*$ are player 1 and 2's security strategies in the primal game $\Gamma_{n,\lambda}(p,q)$, respectively.

Moreover, $w^{1}_{n,\lambda}(\mu^*,q)=w^{2}_{n,\lambda}(p,\nu^*)=0$, where $\mu^*$ and $\nu^*$ is given in equation (\ref{eq: mu star}) and (\ref{eq: nu star}), respectively.
\end{corollary}

\begin{proposition}
\label{prop: recursive formula, dual games}
The value of type 1 dual game satisfies the following recursive formula:
\begin{align}
w^{1}_{n,\lambda}(\mu,q)=& \min_{Y} \min_{\beta_{a\in A, b\in B}} \max_{\Pi} \sum_{a,k} \Pi(a,k) \Big[ \mu(k) + \sum_{l, b} G_{k,l} (a,b)Y(b,l)q(l)+ \nonumber \\  
& \lambda \sum_{b} \bar{y}_{q,Y}(b) \Big( w^{1}_{n-1,\lambda}(\beta_{a,b}, q^+_{q,Y}) - \sum_{k'}P_{a,b}(k,k')\beta_{a,b}(k') \Big)\Big] \label{dual_recursive_type_1} \\  
\doteq & \min_{Y} \min_{\beta} \max_{\Pi} \mathfrak{S}^1_{\mu,q,Y,\beta,\Pi}(w^1_{n-1,\lambda}) \nonumber
\end{align}
where $\Pi(a,k)=Pr(a \cap k)$. The value of type 2 dual game satisfies the following recursive formula:
\begin{align}
w^{2}_{n,\lambda}(p,\nu)=& \max_{X} \max_{\alpha_{a\in A, b\in B}} \min_{\Psi} \sum_{b,l} \Psi(b,l) \Big[ \nu(l) + \sum_{k,  a} G_{k,l}(a,b) X(b,l) p(k)+ \nonumber \\
&  \lambda \sum_{a} \overline{x}_{p,X}(a)\Big( w^{2}_{n-1,\lambda}(p^{+}_{p,X}, \alpha_{a,b})- \sum_{l'} Q_{a,b}(l,l')\alpha_{a,b}(l')\Big)\Big] \label{dual_recursive_type_2} \\
\doteq & \max_{X} \max_{\alpha} \min_{\Psi} \mathfrak{S}^2_{p,\nu,X,\alpha,\Psi}(w^2_{n-1,\lambda}) \nonumber
\end{align}
where $\Psi(b,l)=Pr(b \cap l)$.
\end{proposition}
\begin{proof}
According to equation (\ref{eqn 5}),
\begin{align}
w^{1}_{n,\lambda}(\mu,q)&= \max_{p} p^{T}\mu + v_{n,\lambda}(p,q) \nonumber \\
&= \max_{p} p^{T}\mu + \max_{X} \min_{Y} \Big[ \bar{G}(p,q,X,Y) + \lambda \sum_{a,b} \bar{x}_{p,X}(a) \bar{y}_{q,Y}(b) v_{n-1,\lambda}(p^+,q^+)\Big] \nonumber \\
&= \max_{p} p^{T}\mu + \max_{X} \min_{Y} \Big[\sum\limits_{k,l} p(k)q(l)\sum\limits_{a,b}X(a,k)G_{k,l}(a,b)Y(b,l) + \lambda  \sum_{a,b} \sum\limits_{k} \nonumber \\
& p(k)X(a,k)\overline{y}_{q,Y}(b)v_{n-1,\lambda}\Big(\sum\limits_{k} P_{a,b}(k,k')\frac{p(k)X(a,k)}{\sum_{k} p(k)X(a,k)},q^+\Big)\Big] \nonumber \\
&= \max_{p} \max_{X} \min_{Y} \Big[ \sum_k Pr(k) \mu(k) + \sum\limits_{k,l,a,b} Pr(k)q(l)Pr(a|k)G_{k,l}(a,b)Y(b,l) \nonumber \\
&+ \lambda \sum_{a,b,k} Pr(k)Pr(a|k)\overline{y}_{q,Y}(b) v_{n-1,\lambda}\Big(\sum\limits_{k} P_{a,b}(k,k')\frac{Pr(k) Pr(a|k)}{\sum_k Pr(k) Pr(a|k)},q^+\Big)\Big] \nonumber \\
&= \max_{\Pi} \min_{Y} \Big[ \sum_{k,a} \Pi (a,k)\mu(k) +\lambda \sum\limits_{k,l,a,b} \Pi (a,k) q(l)G_{k,l}(a,b)Y(b,l) +\sum_{a,b,k} \Pi (a,k)  \nonumber \\
&\overline{y}_{q,Y}(b) v_{n-1,\lambda}\Big(\sum\limits_{k} p_{a,b}(k,k') \frac{\Pi (a,k)}{\sum_k \Pi (a,k)},q^+\Big)\Big]; \ \ [Let,\ \Pi(a,k)= Pr(a\cap k)] \\
&= \max_{\Pi} \min_{Y} F(\Pi,Y) \label{define f}
\end{align}
According to statement \ref{Lemma 4.7} we can write,
\begin{align*}
w^{1}_{n,\lambda}(\mu,q)&= \max_{\Pi} \min_{Y} \Big[ \sum_{k,a} \Pi (a,k)\mu(k) + \sum\limits_{k,l,a,b} \Pi (a,k)q(l) G_{k,l}(a,b)Y(b,l) + \lambda \sum_{a,b,k} \\
& \Pi (a,k)\bar{y}_{q,Y}(b) v_{n-1,\lambda}\Big(\sum\limits_{k} P_{a,b}(k,k') \frac{\Pi (a,k)}{\sum_k \Pi (a,k)},q^+\Big)\Big]\\
&= \max_{\Pi} \min_{Y} \Big[ \sum_{k,a} \Pi (a,k)\mu(k) + \sum\limits_{k,l,a,b} \Pi (a,k)q(l) G_{k,l}(a,b)Y(b,l) + \lambda \sum_{a,b}\\
&\bar{y}_{q,Y}(b) v_{n-1,\lambda}\Big(\sum\limits_{k} P_{a,b}(k,k'){\Pi (a,k)},q^+\Big)\Big]\\
&= \max_{\Pi} \min_{Y} \Big[ \sum_{k,a} \Pi (a,k)\mu(k) + \sum\limits_{k,l,a,b} \Pi (a,k) q(l)G_{k,l}(a,b)Y(b,l) + \lambda \sum_{a,b,k} \\
&\Pi (a,k) \bar{y}_{q,Y}(b) v_{n-1,\lambda}\Big(\sum_{k} P_{a,b}(k,k'), q^{+}\Big)\Big]\\
\end{align*}
$ F(\Pi,Y) $ is linear in Y and concave in $ \Pi $ [\ref{Lemma 4.8}]. According to Sion's theorm \cite{sion1958general},
\begin{align*}
w^{1}_{n,\lambda}(\mu,q) &= \max_{\Pi} \min_{Y} F(\Pi,Y)\\
&= \min_{Y} \max_{\Pi} F(\Pi,Y)
\end{align*}
\begin{align*}
&w^{1}_{n,\lambda}(\mu,q)\\
&= \min_{Y} \max_{\Pi}  \sum_{k,a} \Pi (a,k) \Big[ \mu(k) + \sum\limits_{l,b}q(l) G_{k,l}(a,b)Y(b,l) + \lambda  \sum_{b} \bar{y}_{q,Y}(b) v_{n-1,\lambda}(p^+,q^+)\Big]\\
&= \min_{Y} \max_{\Pi}  \sum_{k,a} \Pi (a,k) \Big[ \mu(k) + \sum\limits_{l,b} q(l)G_{k,l}(a,b)Y(b,l) +\lambda  \sum_{b} \bar{y}_{q,Y}(b) \min_{\beta} \big(w^{1}_{n-1,\lambda}(\beta,q)\\
&-{p^+}^{T}\beta \big) \Big]\\
&= \min_{Y} \max_{\Pi} \min_{\beta} \Big[ \sum_{k,a} \Pi (a,k)  \mu(k)+ \sum_{k,a} \Pi (a,k) \sum\limits_{l,b} q(l)G_{k,l}(a,b)Y(b,l)+\lambda \sum_{k,a}\Pi (a,k) \\
&  \sum_{b} \bar{y}_{q,Y}(b) w^{1}_{n-1,\lambda}(\beta,q)- \lambda \sum_{k,a} \Pi (a,k) \sum_{b} \bar{y}_{q,Y}(b)\sum_{k,k'} P_{a,b}(k,k')\frac{p(k)X(a,k)}{\sum\limits_{k}p(k)X(a,k)}\beta(k') \Big]\\
&= \min_{Y} \max_{\Pi} \min_{\beta} \Big[ \sum_{k,a} \Pi (a,k)  \mu(k)+ \sum_{k,a} \Pi (a,k) \sum\limits_{l,b}q(l) G_{k,l}(a,b)Y(b,l)+\lambda \sum_{k,a} \Pi (a,k) \\
&\sum_{b} \bar{y}_{q,Y}(b) w^{1}_{n-1,\lambda}(\beta,q) - \lambda \sum_{k,a} \Pi (a,k) \sum_{b} \bar{y}_{q,Y}(b)\sum_{k,k'} P_{a,b}(k,k')\frac{\Pi (a,k)}{\sum\limits_{k}\Pi (a,k)}\beta(k') \Big]\\
&= \min_{Y} \max_{\Pi} \min_{\beta} \sum_{k,a} \Pi (a,k) \Big[   \mu(k)+ \sum\limits_{l,b}q(l) G_{k,l}(a,b)Y(b,l)+\lambda  \sum_{b} \bar{y}_{q,Y}(b)\Big(w^{1}_{n-1,\lambda}(\beta,q)\\
& -\sum_{a,k'} P_{a,b}(k,k')\beta(k') \Big]\\
&= \min_{Y} \max_{\Pi} \min_{\beta} \Phi(\Pi,\beta)
\end{align*}
$ w^{1}_{n-1,\lambda}(\beta,q) $ is convex in $ \beta $. To prove that we need to show,
\begin{align*}
 w^{1}_{n-1,\lambda}(\beta,q)\leq t w^{1}_{n-1,\lambda}(\beta_{1},q)+ (1-t)t w^{1}_{n-1,\lambda}(\beta_{2},q)
\end{align*}
where, $ \beta= t\beta_1 + (1-t)\beta_2 $ and $ t\in[0,1] $
\begin{align*}
L.H.S &= w^{1}_{n-1,\lambda}(\beta,q)\\
&= w^{1}_{n-1,\lambda}(t\beta_1 + (1-t) \beta_2,q)\\
&= \max_p p^T(t\beta_1 + (1-t) \beta_2) + v_{n-1,\lambda}(p,q)\\
&= \max_p p^Tt\beta_1 + p^T(1-t) \beta_2) + v_{n-1,\lambda}(p,q)+tv_{n-1,\lambda}(p,q)-tv_{n-1,\lambda}(p,q)\\
&= \max_p \Big[t [ p^T \beta_{1}+  v_{n-1,\lambda}(p,q)] +(1-t) [ p^T \beta_{2}+ v_{n-1,\lambda}(p,q)] \Big]\\
R.H.S &= t w^{1}_{n-1,\lambda}(\beta_{1},q)+ (1-t)t w^{1}_{n-1,\lambda}(\beta_{2},q)\\
&= \max_{p} t [ p^T \beta_{1}+ v_{n-1,\lambda}(p,q)] + \max_{p}(1-t)[p^T \beta_{2}+ v_{n-1,\lambda}(p,q)] 
\end{align*}
We know,
\begin{align*}
& \max_p \Big[t [ p^T \beta_{1}+  v_{n-1,\lambda}(p,q)] +(1-t) [ p^T \beta_{2}+ v_{n-1,\lambda}(p,q)] \Big] \leq \max_{p} t [ p^T \beta_{1}+ \\
&v_{n-1,\lambda}(p,q)] + \max_{p}(1-t)[p^T \beta_{2}+ v_{n-1,\lambda}(p,q)] \\
&\therefore w^{1}_{n-1,\lambda}(\beta,q) \leq t w_{n-1,\lambda}(\beta_{1},q)+ (1-t)t w_{n-1,\lambda}(\beta_{2},q)\\
\end{align*}
According to Sion's theorm \cite{sion1958general},
\begin{align*}
w^{1}_{n-1,\lambda}(\mu,q)&= \min_{Y} \min_{\beta} \max_{\Pi}  \Phi(\Pi,\beta)
\end{align*}
By using similar method we can find the recursive formula for the game value of type 2 dual game. \qed
\end{proof}

\begin{corollary}
\label{col: sufficient statistics, dual games}
Player $1$'s security strategy at stage $t$ in type $2$ dual game $\Tilde{\Gamma}^{2}_{N,\lambda}(p,\nu)$ only depends on $p_t$ and $\nu_t$, together with $t$ if $N$ is finite.
Player $2$'s security strategy at stage $t$ in type $1$ dual game $\Tilde{\Gamma}^{1}_{N,\lambda}(\mu,q)$ only depends on $\mu_t$ and $q_t$, together with $t$ if $N$ is finite.
\end{corollary}
\begin{proof}
At stage $1$ in $N$-stage game, player $2$ knows it's initial state $q$ and initial realized vector $\mu$. Using these information it can compute it's optimal strategy $\tau^*_1$ for stage $1$ by solving the optimal problem in equation (\ref{dual_recursive_type_1}). Player $2$ chooses it's stage $1$ action according to the optimal strategy. If player $1$ chooses some action for stage $1$ then the stage $1$ playoff of this game is,
\begin{align*}
&E_{p,q,\sigma, \tau}\Big( \mu(k_1)+ G_{k_1,l_1}(a_1,b_1) \Big)\\
=& \sum_{k_1} \mu(k_1)Pr(k_1) + \sum_{k_1,l_1,a_1,b_1} G_{k_1,l_1}(a_1,b_1) Pr(k_1,l_1,a_1,b_1)\\
=& \sum_{k_1} \mu(k_1)Pr(k_1) + \sum_{k_1,l_1,a_1,b_1} G_{k_1,l_1}(a_1,b_1) Pr(a_1|k_1)Pr(b_1|l_1)Pr(k_1)Pr(l_1)\\
=& \sum_{k,a} \Pi(a,k) \mu(k) + \sum_{k,l,a,b} G_{k,l}(a,b) \Pi(a,k)Pr(b|l)Pr(l)\\
=& \sum_{k,a} \Pi(a,k) \Big[ \mu(k) + \sum_{l,b} G_{k,l}(a,b) Y(b,l)q(l) \Big]
\end{align*}
After stage $1$, player $2$ plays the optimal strategy $\tau^*$ for $N-1$ stage game with initial parameter $q^+$ and $\beta_{a,b}^*$. When player $1$ chooses some action and player $2$  chooses action according to its optimal strategy then payoff of the game $\tilde{\Gamma}^1_{N-1, \lambda}(\beta_{a,b}^*, q^+)$ is, 
\begin{align*}
E_{p^+,q^+, \sigma, \tau^*} \Big( \beta_{a,b}^*(k')+ \sum_{t=1}^{N-1} \lambda^{t-1} G_{k_t,l_t}(a_t, b_t) \Big)
\end{align*}
If player $2$ uses the optimal strategy to take action for stage $2:N$ in the $N$ stage game then the primal game payoff of stage $2:N$ will be,
\begin{align*}
&E \big( E_{p^+,q^+, \sigma, \tau^*} \big(\sum_{t=2}^N \lambda^{t-1} G_{k_t,l_t}(a_t,b_t)\big)| a_1, b_1 \Big)\\
=&  E \Big( \lambda E \big( \beta_{a,b}^*(k')+ \sum_{t=1}^{N-1} \lambda^{t-1} G_{k_t,l_t}(a_t, b_t)-\beta_{a,b}^*(k') \big)| a_1, b_1 \Big)\\
=& \lambda \sum_{a_1,b_1} E\Big(\beta_{a,b}^*(k'))+ \sum_{t=1}^{N-1} \lambda^{t-1} G_{k_t,l_t}(a_t, b_t)-\beta_{a,b}^*(k')\Big) Pr(a_1,b_1)\\
=& \lambda \sum_{a_1,b_1} \Big[ E\big(\beta_{a,b}^*(k')+ \sum_{t=1}^{N-1} \lambda^{t-1} G_{k_t,l_t}(a_t, b_t)\big)-E(\beta_{a,b}^*(k')) \Big] Pr(a_1,b_1)\\
=& \lambda \sum_{a_1,b_1} Pr(a_1,b_1)\Big[ w_{n-1, \lambda}^1(\beta_{a,b}^*,q^+)-\sum_{k'}\beta_{a,b}^*(k') Pr(k') \Big]\\
=& \lambda \sum_{a,b,k,l} Pr(a \cap k)Pr(b|l)Pr(l)\Big[ w_{n-1, \lambda}^1(\beta_{a,b}^*,q^+)-\sum_{k'}\beta_{a,b}^*(k') P_{a,b}(k,k') \Big]\\
=& \sum_{a,k} \Pi(a,k) \lambda \sum_{b} \bar{y}_{q,Y}(b) \Big( w^{1}_{n-1,\lambda}(\beta_{a,b}, q^+_{q,Y}) - \sum_{k'}P_{a,b}(k,k')\beta_{a,b}(k') \Big)
\end{align*}
From the maximization of the total payoff of stage $1$ and $2:N$ we can say, in the worst case player 2 can achieve the game value $w_{n, \lambda}(\mu, q)$. Similarly, we can prove this for player $1$. \qed
\end{proof}

\subsubsection{Concavity of $ F(\Pi,Y) $}
\label{Lemma 4.8}
$ F(\Pi,Y) $ is concave in $ \Pi $, where $F(\Pi,Y)$ is defined as equation (\ref{define f}).
\begin{proof}
According to Lemma \ref{Lemma 4.7}, $v_{n-1,\lambda}(\alpha p,q)= \alpha v_{n-1,\lambda}(p,q)$. Then we can write,
\begin{align*}
F(\Pi,Y)&= \max_{\Pi} \min_{Y} \Big[ \sum_{k,a} \Pi (a,k)\mu(k) + \sum\limits_{k,l,a,b} \Pi (a,k)q(l) G_{k,l}(a,b)Y(b,l) + \lambda \\
&\sum_{a,b,k} \Pi (a,k)\overline{y}_{q,y}(b) v_{n-1,\lambda}\Big(\sum\limits_{k} p_{a,b}(k,k')\frac{\Pi (a,k)}{\sum_k \Pi (a,k)},q^+\Big)\Big]\\
&= \max_{\Pi} \min_{Y} \Big[ \sum_{k,a} \Pi (a,k)\mu(k) + \sum\limits_{k,l,a,b} \Pi (a,k)q(l) G_{k,l}(a,b)Y(b,l) + \lambda \\
&\sum_{a,b} \overline{y}_{q,y}(b) v_{n-1,\lambda}\Big(\sum\limits_{k} p_{a,b}(k,k'){\Pi (a,k)},q^+\Big)\Big]\\
&= \max_{\Pi} \min_{Y} \Big[ \sum_{k,a} \Pi (a,k)\mu(k) + \sum\limits_{k,l,a,b} \Pi (a,k)q(l) G_{k,l}(a,b)Y(b,l) + \lambda \\
&\sum_{a,b,k} \Pi (a,k) v_{n-1,\lambda}\Big(p_{a,b}(k,k'), q^{+}\Big)\Big]
\end{align*}
From this equation we can see that $ F(\Pi,Y) $ is linear in $\Pi$. So  $ F(\Pi,Y) $ is both concave and convex in $\Pi$. The proof is complete. \qed
\end{proof}

\subsubsection{Linearity of $v_n$}
\label{Lemma 4.7}
\begin{align*}
v_{n-1,\lambda}(\alpha p,q)&= \alpha v_{n-1,\lambda}(p,q)\\
v_{n-1,\lambda}(p,\alpha q)&= \alpha v_{n-1,\lambda}(p,q)
\end{align*}
\begin{proof}
For $ n=1 $,
\begin{align*}
v_{n,\lambda}(\alpha p,q)&=\max_{X} \min_{Y} [\sum\limits_{k,l} \alpha p(k)q(l)\sum\limits_{a,b}X(a,k)G_{k,l}(a,b)Y(b,l)]\\
&=\max_{X} \min_{Y}\alpha [\sum\limits_{k,l}  p(k)q(l)\sum\limits_{a,b}X(a,k)G_{k,l}(a,b)Y(b,l)]\\
&= \alpha v_{n,\lambda}( p,q)
\end{align*}
Suppose, for $ n $, $ v_{n,\lambda}(\alpha p,q)= \alpha v_{n,\lambda}( p,q)$. For $ n+1 $,
\begin{align*}
v_{n+1,\lambda}(\alpha p,q)&=\max_{X} \min_{Y} [\sum\limits_{k,l} \alpha p(k)q(l)\sum\limits_{a,b}X(a,k)G_{k,l}(a,b)Y(b,l)+\sum_{a,b} \sum\limits_{k} \alpha p(k)\\
& X(a,k)\overline{y}_{q,y}(b)v_{n,\lambda}\Big(\sum\limits_{k} p_{a,b}(k,k')\frac{\alpha p(k)X(a,k)}{\sum_k \alpha p(k)X(a,k)},q^+\Big)\Big]]\\
&=\max_{X} \min_{Y} \alpha [\sum\limits_{k,l} p(k)q(l)\sum\limits_{a,b}X(a,k)G_{k,l}(a,b)Y(b,l)+\sum_{a,b} \sum\limits_{k} p(k)\\
& X(a,k)\overline{y}_{q,y}(b)v_{n,\lambda}\Big(\sum\limits_{k} p_{a,b}(k,k')\frac{p(k)X(a,k)}{\sum_k p(k)X(a,k)},q^+\Big)\Big]]\\
&=\alpha v_{n+1,\lambda}( p,q)
\end{align*}
So, the proof is complete. \qed
\end{proof}

\bibliographystyle{abbrv}
\bibliography{myref}